%% file: main.tex
\documentclass[11pt]{article}
\usepackage{microtype}
\usepackage{titling}
\usepackage[english]{babel}
\usepackage{blindtext}
\usepackage{tikz}
\usepackage{amsmath,amsthm,amsfonts,bm}
\usepackage{stmaryrd}
\usepackage{xfrac}
\usepackage{url}
\usepackage{multirow}
\usepackage{subcaption}
\usepackage[suppress]{color-edits}
\usepackage[margin=1in]{geometry}
\usepackage[cal=pxtx]{mathalfa}
\usepackage{cleveref}
\usepackage{wrapfig}
\usepackage{thmtools, thm-restate}
\usepackage[]{algorithm2e}
\usepackage{mathrsfs}
\usepackage{enumitem}
\usepackage{authblk}
\declaretheorem{theorem}

\addauthor[Bobby]{rdk}{red}
\addauthor[Daniel]{da}{orange}
\addauthor[Tegan]{tew}{violet}
\addauthor[Vishal]{vs}{blue}
\addauthor[Hakim]{hw}{orange}

\newtheorem{lemma}[theorem]{Lemma}
\newtheorem{corollary}[theorem]{Corollary}

\newtheorem*{thm*}{Theorem}
\theoremstyle{definition}
\newtheorem{dfn}{Definition}
\newtheorem{rmk}{Remark}

\newcommand{\reals}{{\mathbb{R}}}
\newcommand{\eps}{{\varepsilon}}
\newcommand{\E}{\mathbb{E}}
\newcommand{\xvec}{\vec{X}}
\newcommand{\vc}[1]{\bm{#1}}
\newcommand{\trans}{{\mathsf{T}}}

\newcommand{\gpi}{G_{\bm{\pi}}}
\newcommand{\evirt}{E_{\text{virt}}}
\newcommand{\ephys}{E_{\text{phys}}}

\newcommand{\bigotilde}{\tilde{\mathcal{O}}}
\newcommand{\bigo}{\mathcal{O}}
\newcommand{\semih}{g} 
\newcommand{\chframe}{(C,\semih)\text{-constellation}} 

\newcommand{\bsq}[1][q]{{B_{\sigma}^{(#1)}}}
\newcommand{\evnt}{{\mathcal{E}}^q}
\newcommand{\tl}{{\mathrm{tail}}}
\newcommand{\hd}{{\mathrm{head}}}

\newcommand{\rplus}{[0,\infty)}
\newcommand{\pths}{\mathcal{P}}

\newcommand{\floor}[1]{\left\lfloor #1 \right\rfloor}

\newenvironment{lparray}%
{\begingroup  \begin{array}{l@{\hspace{8mm}}l@{\hspace{8mm}}l}}%
{\end{array} \endgroup}

\date{}
\begin{document}


\begingroup
\let\clearpage\relax
\include{author-abstract}

\include{introduction}

\include{definitions}

\include{upper-bound}

\include{tailbound}
\include{semi-orn-design}

\include{lower-bound}

\include{open-questions}

\endgroup

\bibliographystyle{alpha}
\bibliography{biblio}{}

\newpage
\appendix
\include{k-contentious-constellations}
\include{avg-lat-lb}

\end{document}

%% file: author-abstract.tex
\title{Breaking the VLB Barrier for Oblivious Reconfigurable Networks}

 \author[1]{Tegan Wilson}
 \author[1]{Daniel Amir}
 \author[1]{Nitika Saran}
 \author[1]{Robert Kleinberg}
 \author[2]{Vishal Shrivastav}
 \author[1]{Hakim Weatherspoon}

 \affil[1]{Cornell University}
 \affil[2]{Purdue University}

\begin{titlingpage}
\maketitle

\begin{abstract}
In a landmark 1981 paper,
Valiant and Brebner gave birth to the study
of oblivious routing and, simultaneously,
introduced its most powerful and ubiquitous
method: {\em Valiant load balancing (VLB)}.
By routing messages through a randomly sampled
intermediate node, VLB lengthens routing paths
by a factor of two but gains the crucial
property of {\em obliviousness}: it balances
load in a completely decentralized manner,
with no global knowledge of the communication
pattern.
Forty years later, with datacenters handling
workloads whose communication pattern
varies too rapidly to allow centralized
coordination, oblivious routing is as
relevant as ever, and VLB continues to
take center stage as a widely used
--- and in some settings, provably
optimal --- way to balance load in the
network obliviously to the traffic
demands. However, the ability of the
network to rapidly reconfigure its
interconnection topology gives rise
to new possibilities.

In this work we
revisit the question of whether VLB
remains optimal in the novel setting
of reconfigurable networks.
Prior work showed that VLB achieves the
optimal tradeoff between latency and
{\em guaranteed} throughput. In this
work we show that a strictly superior
latency-throughput tradeoff is achievable
when the throughput bound is relaxed
to hold with high probability. The same
improved tradeoff is also achievable with
guaranteed throughput under
time-stationary demands, provided the
latency bound is relaxed to hold with
high probability and that the network is
allowed to be {\em semi-oblivious},
using an oblivious (randomized) connection
schedule but demand-aware routing.
We prove that the latter result is not
achievable by any fully-oblivious
reconfigurable network design,
marking a rare case in which semi-oblivious
routing has a provable asymptotic
advantage over oblivious routing. 
Our results are enabled by a novel
oblivious routing scheme that improves
VLB by stretching routing paths the
minimum possible amount --- an additive
stretch of 1 rather than a multiplicative
stretch of 2 --- yet still manages to
balance load with high probability
when either the traffic demand matrix or
the network's interconnection schedule
are shuffled by a uniformly random
permutation. To analyze our routing
scheme we prove an exponential tail
bound which may be of independent interest,
concerning the distribution of values
of a bilinear form on an orbit of a
permutation group action.

\end{abstract}
\end{titlingpage}

%% file: introduction.tex
\section{Introduction} \label{sec:intro}

Reconfigurable networks use rapidly reconfiguring switches to create a dynamic time-varying topology, allowing for great flexibility in efficiently routing traffic.
This idea has gained prominence due to recent technologies such as optical circuit switching \cite{helios,c-through} and free-space optics \cite{ mirrormirror-2012,firefly-free-space-optics,ProjecToR-free-space-optics} 
that enable reconfigurations within microseconds \cite{microsecond-circuit-switching-dc,reactor-circuit-switching} or even nanoseconds \cite{optical-switching-nanosecond-zehnder-interferometer,nanosecond-hybrid-optical-switches}.
Datacenter network architectures that leverage this capability are now being actively explored, including with recent prototype systems \cite{rotornet,shoal,opera,sirius} and theoretical modeling and analysis \cite{stoc-paper,apocs-paper,mars}. 
The rate of change of datacenter network workloads (summarized by a time-varying traffic demand matrix) has already outpaced the reconfiguration speeds achievable using a central controller~\cite{rotornet}, driving researchers to focus on \textbf{oblivious reconfigurable networks (ORNs)}, which use a \textit{demand-oblivious} reconfiguration and routing mechanism that is fully decentralized.

An analogous set of questions came to the fore in an earlier
era of computing research, when the focus was on designing
communication schemes for parallel computers. The network
model at that time --- a fixed, bounded-degree topology ---
was very different, but the objective was the same: to
efficiently simulate arbitrary communication patterns among
a set of $N$ nodes without requiring any centralized control.
In a landmark 1981 paper, Valiant and Brebner articulated
the central problem in terms that still resonate with the
practice of modern datacenter networking.
\begin{quotation} \begin{em}
    The fundamental problem that arises in simulating
    on a realistic machine one step of an idealistic
    computation is that of simulating arbitrary connection
    patterns among the processors via a fixed sparse network\dots
    For routing the packets the strategy will have to be based
    on only a minute fraction of the total information
    necessary to specify the complete communication pattern.
\end{em} \end{quotation}
The solution proposed by Valiant and Brebner, which
henceforth came to be known as {\em Valiant load balancing}
or {\em VLB}, was beautifully simple:
to send data from source $s$ to destination $t$, sample
an intermediate node $u$ uniformly at random. Then form
a routing path from $s$ to $t$ by concatenating ``direct
paths'' from $s$ to $u$ and from $u$ to $t$. (The definition
of direct paths may depend on the network topology; often
shortest paths suffice.) This lengthens routing paths by
a factor of two and thus consumes twice as much bandwidth
as direct-path routing. However, crucially, it is
{\em oblivious}: the distribution over routing paths
from $s$ to $t$ depends only on the network topology,
not the communication pattern. Oblivious routing schemes
satisfy the desideratum of being
``based on only a minute fraction of the total information
necessary to specify the complete communication pattern''
in the strongest possible sense. 

The focus of oblivious routing research in the 1980's
was on network topologies designed to enable efficient
communication among a set of processors. These
topologies, such as hypercubes and shuffle exchange
networks, tended to be highly symmetric (often with
vertex- or edge-transitive automorphism groups) and
tended to have low diameter and no sparse cuts.
One could loosely refer to this class of networks
as {\em optimized topologies}. A second phase of
oblivious routing research, initiated by R\"{a}cke
in the early 2000's, designed oblivious routing
schemes for {\em general topologies}. Compared to
optimized topologies, the oblivious routing schemes
for general topologies require much greater
overprovisioning, inflating the capacity of
each edge by at least a logarithmic
factor compared to the capacity that
would be needed if routing  could be done
using an optimal (non-oblivious) multicommodity
flow. The construction of oblivious routing schemes
with polylogarithmic~\cite{bienk03,harrel03,raecke02}
and eventually logarithmic~\cite{raecke08}
overhead was a seminal discovery for theoretical
computer science, but did not improve over the
performance of VLB for optimized topologies.

Remarkably, more than 40 years after the
introduction of VLB, it remains the state of
the art for oblivious routing in optimized
topologies. In fact, existing results in the
literature show that the factor-of-two
overprovisioning associated with VLB is
optimal in at least two important contexts:
when building a network of fixed-capacity
links to permit any communication pattern
with bounded ingress and egress rates per
node~\cite{keslassy2005optimal,zhang2005designing,babaioff2007optimality},
and when designing an oblivious reconfigurable
network with bounded maximum latency, again to permit any communication pattern
with bounded ingress and egress rates per node~\cite{stoc-paper}.

Running the network is responsible for a significant fraction of the cost of modern datacenters.
The capital cost of the networking equipment alone accounts for around 15\% of the total cost to build and run a datacenter; this increases to over 30\% when including indirect costs such as power and cooling for network equipment \cite{greenberg2008cost,basu2020architecture}.
Overprovisioning the network increases these costs proportionally \cite{shoal},
which motivates investigating when it is possible to ``break the VLB barrier'' and reap the benefits of oblivious routing without paying the cost of provisioning twice as much capacity as needed for optimal demand-aware routing.

In this work we show that {\em the ability to
  randomize the network topology in reconfigurable networks
  indeed allows oblivious routing schemes that break the
  VLB barrier}. We present a novel oblivious routing scheme
for reconfigurable networks with a randomized
connection schedule. The routing paths used by
our scheme exceed the length of shortest
(latency-bounded) paths by the smallest
possible amount: {\em an additive
  stretch of 1 rather than a multiplicative
  stretch of 2}.
Building upon this new routing scheme, we
obtain reconfigurable network designs that
improve the throughput achievable within a
given latency bound by nearly a factor of two,
under two relaxations of obliviousness:
\begin{enumerate}[itemsep=2pt,partopsep=2pt,topsep=6pt]
\item when the network is allowed a small probability
  of violating the throughput guarantee; or
  \label{contrib:whp}
\item   \label{contrib:semi-obliv}
  when the throughput guarantee must
  hold with probability 1, but routing is
  only {\em semi-oblivious}.
\end{enumerate}
Semi-oblivious routing refers to routing
schemes in which the network designer must
pre-commit (in a demand-oblivious manner)
to a limited set of routing paths
between every source and destination, but
the decision of how to distribute flow over
those paths is made with awareness of the
requested communication pattern. In the
context of reconfigurable networks, this
means that the connection schedule is
oblivious but the routing scheme may be
demand-aware. In fact, the semi-oblivious
routing scheme that we refer to in
Result~\ref{contrib:semi-obliv} above
is demand-aware in only a very limited
sense: it uses the oblivious routing scheme
from Result~\ref{contrib:whp} with high
probability, but in the unlikely event
that this leads to congestion on one
or more edges, it reverts to using a
different oblivious routing scheme
that is guaranteed to avoid congestion 
at the cost of incurring higher latency.
Note that this semi-oblivious routing
scheme only requires network nodes to
share one bit of common knowledge about the
communication pattern (namely, whether or not
there exists a congested edge), hence it still
obeys Valiant and Brebner's desideratum
that routing decisions are based on only a
minute fraction of the total information
needed to specify the communication pattern.
In \Cref{sec:semi-obliv-design} we prove that
purely oblivious reconfigurable network
designs (even with a randomized connection
schedule) cannot achieve the same result
as our semi-oblivious design: if the
throughput guarantee must hold with
probability 1, then the average latency
must be strictly asymptotically greater
for oblivious reconfigurable networks
than for semi-oblivious ones. 

\subsection{Summary of results and techniques}
\label{sec:intro-results}

In our abstraction of a reconfigurable network,
a fixed set of $N$ nodes communicates over a
sequence of discrete time steps. In one time
step, each node is allowed to send data
to only one other node and to
receive data from only one\footnote{%
  More generally one could impose a degree
  constraint, $d$, on the number of nodes
  to/from which one node can send/receive
  data in a single time step. Networks with
  degree constraint $d$ can be simulated by
  networks with degree constraint 1, up to
  a slow-down by a factor of $d$~\cite{stoc-paper}.
  Hence, the assumption that $d=1$ is essentially
  without loss of generality, and we will continue
  adopting this assumption throughout the paper.
}
other node. This time-varying connectivity
pattern, called the {\em connection schedule},
may be randomized, but it must be predetermined
in a demand-oblivious manner. To route messages 
through the network, nodes may forward
data over links when they are available in the
connection schedule, and they may buffer messages
when the next link of the designated routing path
is not yet available. The choice of routing paths
is called the {\em routing scheme}.
We allow data to be
fractionally divided over routing paths
(modeling the operation of randomly sampling
one path per data packet) so the routing scheme
is represented by specifying a fractional flow
for each source-destination pair, at each
time step. In an {\em oblivious} reconfigurable
network this flow is predetermined, up to
scaling, in a demand-oblivious manner.
In a {\em semi-oblivious} reconfigurable
network only the connection schedule is
oblivious; the routing scheme may be
demand-aware. 
 
To place our results in context, it helps to
reason a bit about the fundamental limits of
communication in reconfigurable networks.
\begin{enumerate}
\item {\bf Throughput is bounded by
  the inverse of average hop-count.} 
  A network design is said to have throughput $r$
  if it is able to serve any
  communication pattern whose ingress
  and egress rates, at each node in each time step,
  are bounded by $r$ times the
  amount of data that may be transmitted on any link
  per time step. Adopting units in which link capacities
  equal 1, the total amount of demand originating in any
  time step is $r N$ and the total link capacity is $N$. 
  If the average routing path is composed
  of $\semih$ network hops, then the $r N$ units of demand
  originating in any time step will consume $\semih r N$
  units of capacity on average, hence $\semih r \leq 1$.
  Guaranteeing throughput $r$ therefore requires
  guaranteeing average hop-count at most $1/r$.
\item {\bf Hop-count $\semih$ requires latency $L=\Omega(\semih N^{1/\semih})$.}
  A routing path originating at a given node is uniquely
  determined by the set of time steps at which the
  path traverses network hops. (This is because the
  connection schedule specifies a {\em unique} node
  that is allowed to receive messages from any given
  node at any given time.) Hence, in order for any node
  to be able to reach any other node within $L$ time steps
  using a routing path of $\semih$ or fewer hops, it must be
  the case that $\sum_{i=0}^{\semih} \binom{L}{i} \geq N$.
  The solution to this inequality is $L = \Omega(\semih N^{1/\semih})$.
  A more complicated counting argument, which we omit,
  establishes the same lower bound on {\em average}
  latency, even if the bound of $\semih$ hops per path
  is relaxed to hold only on average.
\end{enumerate}
These considerations establish a sort of {\em speed-of-light
barrier} for reconfigurable networking. Even without
the constraint of obliviousness, delivering messages within
$\bigo (\semih N^{1/\semih})$ time steps on average requires $\semih$-hop paths, hence
limits throughput to $1/\semih$. Oblivious or semi-oblivious
network designs can thus be evaluated in relation to this
benchmark.
\begin{figure}[ht]
  \begin{tabular}{|c|c|c|c|}
    \hline
        {\bf Goal} & {\bf Average hop-count} &
        {\bf Throughput} & {\bf Reference} \\
        \hline
        Minimize network hops & $\semih$ & --- & na\"{i}ve counting \\
        \hline
        Uniform multicommodity flow & $\semih$ & $\frac{1}{\semih}$ &
        \cite{stoc-paper} \\
        \hline
        Oblivious routing (w.h.p.) &
        $\semih + 1$ &
        $\frac{1}{\semih+1} - \delta \; \forall \delta > 0$ &
        this work \\
        \hline
        Semi-oblivious routing (prob.\ 1) &
        $\semih + 1 - o(1)$ &
        $\frac{1}{\semih+1} - \delta \; \forall \delta > 0$ &
        this work \\
        \hline
        Oblivious routing (prob.\ 1) &
        $2\semih$ & $\frac{1}{2\semih}$ &
        \cite{stoc-paper} \\
        \hline
  \end{tabular}
  \caption{Bounds for reconfigurable networking with
    average latency constrained by $L = \bigotilde(\semih N^{1/\semih})$.}
  \label{fig:compare-bounds}
\end{figure}
\Cref{fig:compare-bounds} presents a comparison
of bounds for various reconfigurable networking
goals, standardizing on an average latency
constraint of $L = \bigotilde(\semih N^{1/\semih})$ where $\semih$
could be any positive integer (fixed, independent
of $N$). As noted above, even if we ignore
capacity constraints and connect all source-destination
pairs using the minimum number of network hops
subject to this latency constraint, average
path length $\semih$ is unavoidable. Optimal (demand-aware)
routing schemes for the uniform multicommodity
flow match this bound, whereas
optimal oblivious routing schemes require
average path length $2\semih$~\cite{stoc-paper}.
The routing schemes presented in this paper
have average path length $\semih+1$ (minus a $o(1)$ in the case of
semi-oblivious routing), matching the
``speed-of-light barrier'' to within
an additive 1. We also present lower
bounds establishing that this result
is the best possible.

The formal statement 
of our main results generalizes the
foregoing discussion by allowing the
target throughput rate to be any
number (fixed, independent of $N$)
in the interval $(0,\frac12]$. 
\begin{restatable}{theorem}{thmtradeoff}\label{thm:tradeoff}
Given any fixed throughput value $r\in(0,\frac{1}{2}]$, let $\semih = \semih(r) = \lfloor\frac{1}{r}-1\rfloor$ and $\eps = \eps(r) = \semih + 1 - (\frac{1}{r}-1)$, and let 
  \begin{align}
    \label{eq:lupp}
    L_{upp}(r,N) & = \semih N^{1/\semih}  \\
    \label{eq:llow}
	L_{low}(r,N) & = \semih \left((\eps N)^{1/\semih} + N^{1/(\semih+1)}\right) 
  \end{align}
  Assuming $\eps \neq 1$:
  \begin{enumerate} 
  \item \label{to.orn.whp} there exists a family of distributions over ORN designs for infinitely many network sizes $N$ which attains maximum latency $\bigotilde(L_{upp}(r,N))$, and achieves throughput $r$ with high probability;
  \item \label{to.orn.unif} for infinitely many network sizes, there exists a single, fixed ORN design that attains maximum latency $\bigotilde(L_{upp}(r,N))$, and achieves throughput $r$ with high probability over the uniform distribution on permutation  demands;
  \item \label{to.sorn.whp} there exists a family of distributions over semi-oblivious reconfigurable network designs for infinitely many network sizes $N$ which attains maximum latency $\bigotilde(L_{upp}(r,N))$ with high probability (and in expectation) over time-stationary demands, and achieves throughput $r$ with probability 1;
	\item \label{to.orn.whp.lb} furthermore, any fixed ORN design $\mathcal{R}$ of size $N$ which achieves throughput $r$ with high probability over time-stationary demands must suffer at least $\Omega(L_{low}(r,N))$ maximum latency. 
\end{enumerate}
\end{restatable}

The upper and lower bounds on lines~\eqref{eq:lupp}-\eqref{eq:llow}
match to within a constant factor for most values of $r$: when
$ \frac{1}{r} \not\in \bigcup_{m = 2}^{\infty} \left( m - \frac{2}{2^m}, m \right] $
then $\eps \geq 2^{-g}$, so $L_{low} \geq \frac12 L_{upp}$.
The latency of our reconfigurable network designs is
$L_{upp} \cdot \bigotilde( \log N)$, 
hence the upper and lower
bounds in \Cref{thm:tradeoff} agree within a
$\bigotilde(\log N)$
factor for most values of $r$.
See \Cref{fig:semi-v-fully-obliv} for a visualization
of these bounds.
Additionally, like in \cite{apocs-paper} we condition against $\eps= 1$. 
This is due to requiring a strictly positive slack factor between the throughput $r$ and $\frac{1}{g+1}$.

\begin{figure}[h]
 \centering
 \includegraphics[width=0.7\columnwidth]{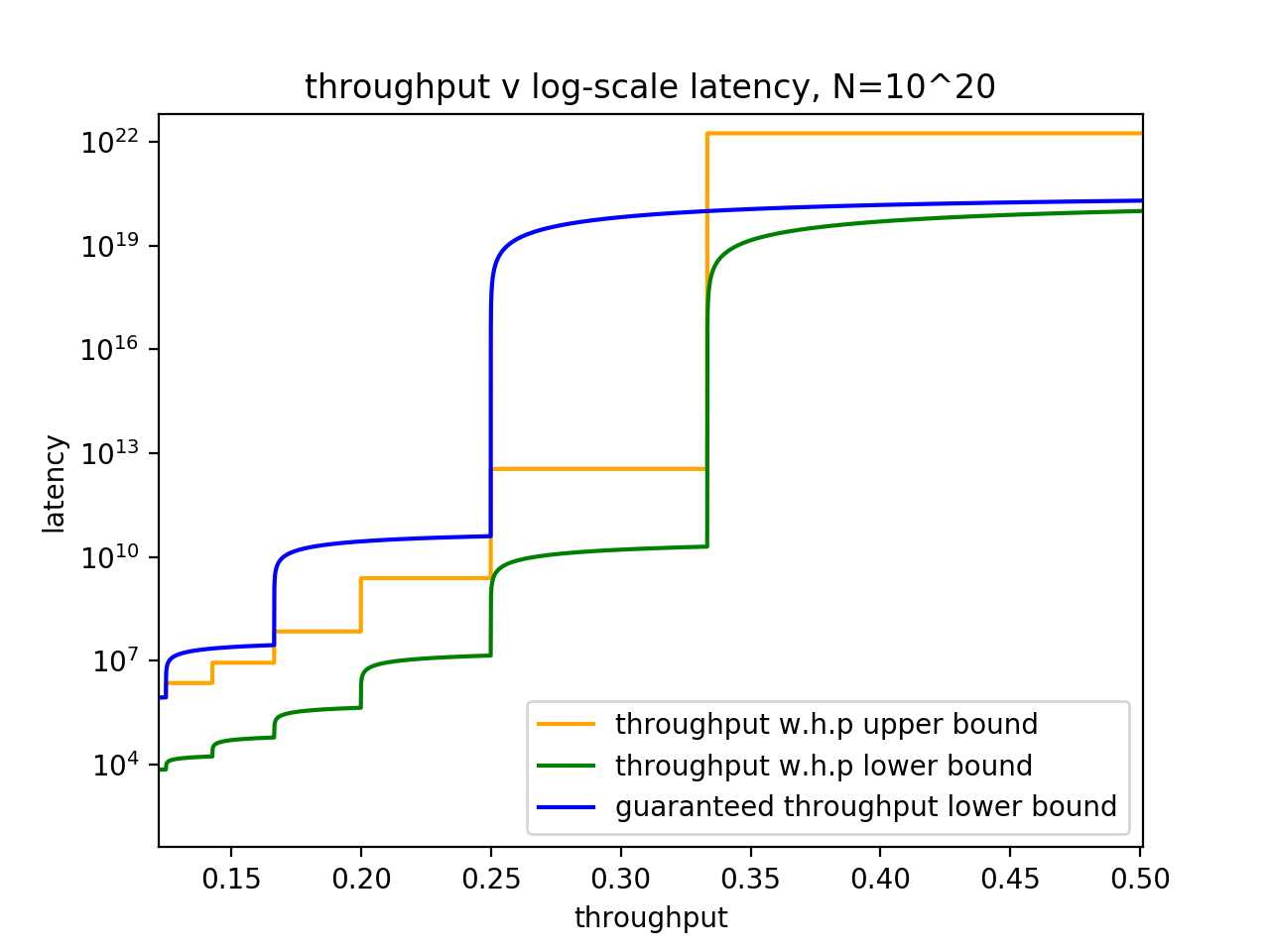}
 \caption{Throughput versus log-scale maximum latency tradeoff curves $\bigotilde(L_{upp})$ and $L_{low}$, when compared against the guaranteed throughput lower bound of \cite{stoc-paper}, on an ORN containing $10^{20}$ nodes.
 }
 \label{fig:semi-v-fully-obliv}
\end{figure}

We conclude this section by sketching how our
routing scheme differs from VLB, and how we
analyze it to obtain the bounds stated above.
Both schemes construct routing paths
composed of {\em spraying hops}, which
transport messages from the source to a random
intermediate node, and {\em direct hops}, which
deliver messages from the intermediate node
to the destination.

In both cases the analysis of the routing
scheme entails showing that the spraying
hops and the direct hops distribute load
evenly over the network links, whenever the
routing scheme is used to serve a {\em permutation
  demand}: a communication pattern where each
source node $s$ seeks to communicate at rate $r$
with a single destination $\sigma(s)$, and the
function $\sigma$ is a permutation of $[N]$. 
For VLB this is easy: intermediate nodes are
sampled uniformly at random, so the distribution
of (source, intermediate node) pairs and the
distribution of (intermediate node, destination)
pairs are both uniform over the set of all pairs of
nodes in the network; a symmetry argument then
suffices to conclude that both the spraying
hops and the direct hops distribute load evenly
over all links. 

In our routing scheme, routing paths consist of
just one spraying hop followed by a direct path
to the destination. Thus, the intermediate node
must be either the source itself, or
one of the nodes reachable by a direct
link from the source node during the first $L$
time steps after the message originates. For
$L < N-1$ it is impossible for the intermediate
node to be uniformly distributed, conditional
on the source. Consequently
(intermediate node, destination)
pairs in our routing scheme are also not uniformly
distributed. This non-uniform distribution
retains some dependence on the permutation
$\sigma$ that associates sources with destinations. 
Hence it is unclear how to guarantee that
for every permutation $\sigma$, 
flow traveling on direct hops will be uniformly
distributed over the edges of the network.

Our main innovation lies in the way we construct
a connection schedule and routing scheme to
ensure (approximately) uniform distribution of
load over edges. The use of a single
spraying hop inevitably reduces the amount
of randomness in the conditional distribution
of the intermediate node given the source,
and we must find a way to regain the lost
randomness without adding extra spraying hops.
To do so, we exploit a novel source of
randomness: we randomize the {\em timing} of
the direct hops.
Prior work~\cite{stoc-paper}
had used a connection schedule based on identifying
the node set $[N]$ with a vector space over a
finite field, and associating time steps with
scalar multiples of the elementary basis vectors.
To each pair of nodes one could then associate a
direct path corresponding to the (unique)
representation of the difference of the node
identifiers as a linear combination of elementary
basis vectors. Thus, the timing of direct hops
was uniquely determined, given the location
of the intermediate node. 

In our connection schedule we
again identify $[N]$ with a vector space over a
finite field. However, there are two key
differences.
First, in some of our designs, the identification
of $[N]$ with a finite vector space is done using
a uniformly random one-to-one correspondence.
This allows us to reduce the analysis of our
(randomized) connection schedule to average-case
analysis of a fixed connection schedule, when
the demand matrices are conjugated by
  a uniformly random permutation matrix.
Second, and more importantly, rather than defining
the connection schedule using a basis of this
vector space, we use an overcomplete system
of vectors which we call a {\em constellation}.
Constellations in a $g$-dimensional vector space
have the property that every $g$-element subset
forms a basis. (In other words, they represent
the uniform matroid of rank $g$.)
Our routing
scheme constructs direct paths between
two nodes by sampling a random $g$-element
subset of the constellation, representing
the difference between the nodes' identifiers
as a linear combination of those $g$ vectors,
and using the corresponding
$g$ time steps of the connection
schedule to form the direct path. 

To show that this method
distributes load approximately uniformly
over edges, we decompose the
load on any given edge as a
sum of $g+1$ random variables,
each of which can be interpreted as 
a bilinear form evaluated on a pair of vectors
representing the number of paths from
each source node to the tail of the
given edge, and from the head of the given
edge to each destination node. The pair
of vectors is sampled at random from
an orbit of the permutation group $S_N$,
which acts on pairs of vectors
either by permuting the coordinates of one
of them (in the case when we're analyzing
a uniformly random permutation demand) or by
permuting the coordinates of both simultaneously
(in the case when we're identifying the node set
with a vector space using a random bijection). 
In both cases, we prove an exponential
tail bound for the value of the bilinear form
on a vector pair randomly sampled from the
permutation-group orbit.
When the permutation acts on only one
element of the ordered pair, the relevant
exponential tail bound follows easily from
the Chernoff bound for negatively associated
random variables~\cite{dubhashi1996balls}.
When the permutation acts on both vectors
simultaneously, the negative association property
does not hold and we take a more indirect
approach, using a 3-coloring of the node
set $[N]$ to decompose the bilinear form
into three parts, each of which can be shown
to satisfy an exponential tail bound after a
suitable conditioning. We believe the resulting
exponential tail bound for bilinear forms may
be of independent interest. 

To improve the high-probability bound
on throughput to a bound that holds
with probability 1, we adopt a semi-oblivious
routing scheme that is a hybrid of a
{\em primary scheme} identical to the
oblivious scheme sketched above,
and a {\em failover scheme} which is
also oblivious, to be used 
in the (low-probability) case that the
primary scheme produces an infeasible
flow. The failover scheme has latency
$\bigotilde(N)$ and
resembles VLB, distributing flow
over two-hop paths from the
source to the destination by
routing through an intermediate
node sampled from a nearly-uniform
distribution.
The challenge is to
modify the connection schedule to
ensure that enough two-hop paths
exist between every source and
destination. We accomplish this
by using a time-varying constellation
in place of the fixed constellation
used by the routing scheme sketched
above. The time-varying sequence of
constellations that we construct forms
a sort of combinatorial design, covering
every vector with non-zero coordinates
an equal number of times. This equal-coverage
property is the key to proving the
that the failover routing scheme balances
load evenly.

\textbf{Our lower bound.} Our lower bound proof is heavily inspired by the lower bound proof of \cite{stoc-paper}.
We build a family of $N!$ linear programs, one for each permutation on the node set, that each maximize throughput subject to a maximum latency constraint $L$. We then take the dual, find a good dual solution, and analyze the objective value of each dual solution.
We then bound the expected objective value across the whole set, and use this to bound the achievable throughput with high probability.
Interestingly, this lower bound result also applies to the guaranteed throughput rate of semi-oblivious designs -- where the connection schedule must be pre-committed to, but the routing algorithm may be adaptive with respect to traffic.

\subsection{Related work} \label{sec:related-works}
\input{related-work}

%% file: related-work.tex
The most important related works, \cite{stoc-paper, apocs-paper}, 
are summarized above in \Cref{sec:intro}.

\textbf{Oblivious routing in general networks.}
Extensive theoretical work in oblivious routing considers the competitive ratio in congestion achievable in general networks, when compared to an adaptive optimal routing.
\cite{raecke02} proved the existence of a poly$\log n$-competitive algorithm for this problem, the competitive ratio later improved upon by \cite{harrel03}.
\cite{bienk03,harrel03,azar03} then developed poly-time algorithms to achieve this result.
Later, these algorithms were implemented and tested in wide-area networks \cite{applegate04}.
\cite{raecke08} further improved to a $\log n$-competitive oblivious routing scheme, based on multiplicative weights and FRT's randomized approximation of general metric spaces by tree metrics \cite{frt04}.
This improved algorithm was again demonstrated in wide-area networks by \cite{smore18}.

Some works add additional constraints to this problem. For example, \cite{gupta06} found a poly$\log n$-competitive routing scheme oblivious to both traffic and the cost functions of edges, and
\cite{hop-constr-obliv} finds a poly$\log n$-competitive ratio when constraining the number of physical hops in paths that both the oblivious routing scheme, and the adaptive benchmark, can use. They also give an algorithm to achieve this.
These works 
assume a fixed graph topology, while our work aims to co-design a network topology and routing scheme.
They also examine congestion, a related but not analogous measure to our definition of throughput, make a guaranteed bound on that congestion instead of a probabilistic bound, and (with the exception of \cite{hop-constr-obliv}) make little attempt to bound latency.

\textbf{Randomized Oblivious Routing.}
There is also extensive work focused on oblivious routing with randomness. This problem is often focused on packet routing, and aims to obliviously choose a single path to route traffic on. It is well known that any such deterministic oblivious routing on a graph of degree $d$ suffers $\Omega(\sqrt{N}/d)$ congestion from an adversarial permutation demand. \cite{kaklam91,Borodin1985Routing}.
Valiant tackles this problem with Valiant Load Balancing, a randomized technique which gives a $\log n$-expected congestion bound on the $d$-dimensional hypercube, butterfly, and mesh networks \cite{vlb-valiant81,valiant82}. He later provided a lower bound in these contexts \cite{valiant83}.
A similar procedure is used in ROMM routing in the hypercube, which selects a larger number of intermediate nodes within the sub-cube containing both the source and destination, and trades off load balancing with latency \cite{romm-routing1, romm-routing2}.
These works differ from ours in that they aim to route discretized packets on paths, and look at the congestion that occurs from worst-case traffic.

\cite{bit-serial-routing91} showed that in bit-serial routing, any random oblivious algorithm on a poly$\log$ degree network requires $\bigo(\log^2 n/\log\log n)$ bit-steps with high probability for almost all permutation traffic, assuming $\log n$-bit messages, extending the Borodin-Hopcroft bound for deterministic algorithms. 
\cite{smore18} examines a partially adaptive (or, semi-oblivious) routing, in which the router precommits to a set of $\log N$ paths between each pair of vertices, and at runtime may only send flow on one of the precommitted paths. This approach was later shown to be poly$\log n$-competitive by \cite{sparse-semi-obliv}. 
Since oblivious routing under the same sparsity constraint cannot be poly$\log n$-competitive, this constitutes an asymptotic separation between the power of semi-oblivious and oblivious routing. To the best of our knowledge~\cite{sparse-semi-obliv} constitutes the first provable asymptotic separation between semi-oblivious and oblivious routing in the literature, and the separation that we prove in \Cref{sec:semi-obliv-provable-sep} is the second such result.

A work that closely models the problem we ask \cite{hajiaghayi-mohammad}, gives a $O(log^2 n)$-competitive algorithm with high probability over random demands in directed graphs, and showed that one cannot do better than $O(\log n/\log\log n)$-competitive with any constant probability. 
Like in non-randomized oblivious routing, they also assume a fixed graph topology, and do not attempt to bound latency.

\textbf{ORN Proposals.}
Although \cite{stoc-paper} is first to name the ORN paradigm, it was used earlier in proposed network architectures and designs.
Rotornet \cite{rotornet} and Sirius \cite{sirius} both use optical circuit switches to build a reconfigurable fabric, and Shoal \cite{shoal} uses electronic circuit switches. These works demonstrate different ways to implement ORNs using physical hardware, however they all use similar connection and routing schedules that maximize throughput, at the expense of latency.
Opera \cite{opera} combines the ORN paradigm with lengthened time slots, high node degrees, and some adaptive routing. This allows a separation into two traffic classes, low-latency and throughput-sensitive. However the design makes significant assumptions about the traffic workload, limiting its flexibility.
Cerberus \cite{cerberus} uses a modification of Rotornet as one component of an optical datacenter network, along with demand-aware reconfiguration and static graphs.

\cite{mars} used the degree of the time-collapsed connection schedule, or \textit{emulated graph}, of an ORN design to bound its throughput, latency, and buffer requirement.
Using these results, the authors derived a formula for the ideal degree $d$ to use for the emulated graph in order to maximize throughput in a buffer-constrained network.
The authors proposed MARS, an ORN design that emulates a de Bruijn graph with this ideal degree to achieve near-optimal throughput under buffer constraints, and evaluated this design through simulation.

%% file: definitions.tex
\section{Definitions} \label{sec:definitions}

\begin{dfn} \label{def:connection-schedule}
A {\em connection schedule} of $N$ nodes and period length $T$ is a sequence of permutations $\bm{\pi}=\pi_0,\pi_1,\ldots,\pi_{T-1}$, each mapping $[N]$ to $[N]$. $\pi_k(i) = j$ means that node $i$ is allowed to send one unit of flow to node $j$ during any timestep $t$ such that $t \equiv k \pmod{T}$.

 The {\em virtual topology} of the connection schedule $\bm{\pi}$ is a directed graph $\gpi$ with vertex set $[N] \times \mathbb{Z}$.
 The edge set of $\gpi$ is the union of two sets of edges, $\evirt$ and $\ephys$. 
 $\evirt$ is the set of {\em virtual edges}, which are of the form $(i,t)\to(i,t+1)$ and represent flow waiting at node $i$ during the timestep $t$. 
 $\ephys$ is the set of {\em physical edges}, which are of the form $(i,t)\to(\pi_t(i),t+1)$, and represent flow being transmitted from $i$ to $\pi_t(i)$ during timestep $t$.
\end{dfn}

We interpret a path in $\gpi$ from $(a,t)$ to $b$ as a potential way to transmit one unit of flow from node $a$ to node $b$, beginning at timestep $t$ and ending at some timestep $t^\prime>t$.
Let $\pths(a,b,t)$ denote the set of paths in $\gpi$ starting at the vertex $(a,t)$ and ending at some $(b,t^\prime)$ for any $t^\prime>t$, and let $\pths_L(a,b,t)$ be the set of such paths for which $t^\prime-t\leq L$. Finally, let $\pths = \bigcup_{a,b,t} \pths(a,b,t)$ denote the set of all paths in $\gpi$.

\begin{dfn} \label{def:flow}
A {\em flow} is a function $f : \pths \to \rplus$. For a given flow $f$, the amount of flow traversing an edge $e$ is defined as:
\[
  F(f,e) =
  \sum_{P \in \pths} f(P) \cdot \bm{1}_{e \in P}
\]
We say that $f$ is {\em feasible} if for every
 physical edge $e \in \ephys$, $F(f,e) \leq 1$. Note that in our definition of feasible, we allow virtual edges to have unlimited capacity.
\end{dfn}

\begin{dfn}\label{def:routing-scheme}
 An {\em oblivious routing scheme} $R$ is a set of functions $R(a,b,t):\pths\rightarrow[0,1]$, one for every tuple $(a,b,t) \in [N]\times[N]\times\mathbb{Z}$, such that:
 \begin{enumerate}
 \item For all $(a,b,t) \in [N]\times[N]\times\mathbb{Z}$,
   $R(a,b,t)$ is a probability distribution supported on
   $\pths(a,b,t)$. 
    \item $R$ has period $T$. In other words, $R(a,b,t)$ is equivalent to $R(a,b,t+T)$ (except with all paths transposed by $T$ timesteps).
  \end{enumerate}

\end{dfn}

\begin{dfn}
   An {\em Oblivious Reconfigurable Network (ORN) design} $\mathcal{R}$ consists of both a connection schedule $\pi_k$ and an oblivious routing scheme $R$.
\end{dfn}

\begin{dfn} \label{def:adapt-rout-scheme}
	A {\em demand-aware routing scheme} $\{S_\sigma : \sigma \mbox{ permut on } [N] \}$ is a set of functions $S_\sigma(a,t):\pths\rightarrow[0,1]$, one for every tuple $(a,t) \in [N]\times\mathbb{Z}$ and permutation $\sigma$ on $[N]$, such that:
  \begin{enumerate}
  \item for all $(a,t,\sigma) \in [N] \times \mathbb{Z} \times S_N$,
    $S_\sigma(a,t)$ is a probability distribution supported on $\pths(a,\sigma(a),t)$. 
    \item $S_\sigma$ has period $T$. In other words, $S_\sigma(a,t)$ is equivalent to $S_\sigma(a,t+T)$ (except with all paths transposed by $T$ timesteps).
  \end{enumerate}
\end{dfn}

\begin{dfn} \label{def:semi-obliv-design}
	A {\em Semi-Oblivious Reconfigurable Network (SORN) Design} $\mathcal{S}$ consists of a connection schedule $\pi_k$ and a demand-aware routing scheme $\{S_\sigma : \sigma \mbox{ permut on } [N] \}$.
\end{dfn}

\begin{dfn} \label{def:latencies}
	 The {\em latency} $L(P)$ of a path $P$ in $\gpi$ is equal to the number of edges it contains (both virtual and physical).
	Traversing any edge in the virtual topology (either virtual or physical) is equivalent to advancing in time by one timestep, so the number of edges in a path equals the elapsed time.
	For an ORN Design $\mathcal{R}$ or SORN design $\mathcal{S}$, the {\em maximum latency} is the maximum over all paths $P$ which may route flow. 
	\begin{align*}
		L_{max}(\mathcal{R}) & = \max_{P\in\pths} \{ L(P) : \exists a,b,t \mbox{ for which } R(a,b,t,P) > 0\} \\
		L_{max}(\mathcal{S}) & = \max_{P\in\pths} \{ L(P) : \exists a,t,\sigma \mbox{ for which } S_\sigma(a,t,P) > 0\}
	\end{align*}
	
	The {\em average (or normalized) latency} is the weighted average across all possible demand pairs and all paths $P$ which may route flow. 
	\begin{align*}
		L_{avg}(\mathcal{R}) & = \frac{1}{N^2T} \sum_{a,b,t}\sum_{P\in\pths(a,b,t)} R(a,b,t,P) L(P) \\
		L_{avg}(\mathcal{S}) & = \frac{1}{NTN!} \sum_{\sigma,a,t}\sum_{P\in\pths(a,\sigma(a),t)} S_\sigma(a,t,P) L(P)
	\end{align*}
\end{dfn}

\begin{dfn} \label{def:demand-functions}
   A {\em demand matrix} is an $N\times N$ matrix which associates to each ordered pair $(a,b)$ a rate of flow to be sent from $a$ to $b$. 
   A {\em demand function} $D$ is a function that associates to every $t \in \mathbb{Z}$ a demand matrix $D(t)$ representing the amount of flow $D(t,a,b)$ originating between each source-destination pair at timestep $t$.
    
   \tewedit{A {\em time-stationary demand} is a demand function in which every demand matrix $D(t)$ is the same.}
   A {\em permutation demand} $D_\sigma$ is a demand function in which every demand matrix is the permutation matrix defined by $\sigma:[N]\rightarrow[N]$. 
   \tewedit{Note that permutation demands are also time-stationary.}
\end{dfn}

\begin{dfn} \label{def:induced-flow}
  If $R$ is an oblivious routing scheme and $D$ is a demand function, the
  {\em induced flow} $f(R,D)$ is defined by:
  \[ f(R,D) = \sum_{(a,b,t) \in [N] \times [N] \times \mathbb{Z} } D(t,a,b) R(a,b,t). \]
  If $\{S_\sigma : \sigma \mbox{ permut on } [N] \}$ is a demand-aware routing scheme and $D_\sigma$ is a permutation demand function (possibly scaled by some constant), then the induced flow is defined by $f(S_\sigma, D_\sigma)$.
\end{dfn}

\begin{dfn} \label{def:guaranteed-thr}
	An ORN Design $\mathcal{R}$ {\em guarantees throughput $r$} if the induced flow $f(R, rD)$ is feasible whenever for all $t$, the row and column sums of $D(t)$ are bounded above by $1$. (Such matrices $D(t)$ are called \emph{doubly sub-stochastic}.)
	An ORN Design $\mathcal{R}$ guarantees throughput $r$ {\em with respect to time-stationary demands} if for every time-stationary demand function $D$ with row and column sums bounded by $1$, then the induced flow $f(R, rD)$ is feasible.
	An easy application of the Birkhoff-von Neumann Theorem establishes the following: in order for an ORN design to guarantee throughput $r$ with respect to time-stationary demands, it is necessary and sufficient that it guarantee throughput $r$ with respect to permutation demands.

	An SORN design $\mathcal{S}$ {\em guarantees throughput $r$} (with respect to permutation demands) if, for every permutation demand $D_\sigma$, the induced flow $f(S_\sigma, rD_\sigma)$ is feasible for all $t$. 
\end{dfn}

\begin{dfn} \label{def:high-prob-throughput}
  A distribution over ORN designs $\mathscr{R}$, is said to {\em achieve throughput $r$ with high probability} if, for any $d \geq 1$ and demand function $D$ such that $D(t)$ is doubly sub-stochastic for all $t$, routing $r D$ on a random $\mathcal{R} \sim \mathscr{R}$ induces a feasible flow with probability at least $1-C_d/N^d$, where $C_d$ is a constant that may depend on $d$.
    
  Similarly, $\mathscr{R}$ is said to {\em achieve throughput $r$ with high probability under the uniform distribution on permutation demands} if, for uniformly random permutations $\sigma$ and any $d\geq 1$, the induced flow $f(R,rD_\sigma)$ is feasible with probability at least $1-C_d/N^d$, where $C_d$ is a constant that may depend on $d$, and the randomness is over both the draw of $\mathcal{R}$ from $\mathscr{R}$ and the draw of $\sigma$ from the uniform distribution over permutations. In the special case when $\mathscr{R}$ is a point-mass distribution on a singleton set $\{\mathcal{R}\}$, we say that the fixed design $\mathcal{R}$ achieves throughput $r$ with high probability under the uniform distribution over permutation demands. 
\end{dfn}

\begin{dfn} \label{def:high-prob-latency}
	 \tewedit{A distribution over SORN designs $\mathscr{S}$, is said to {\em achieve maximum latency $L$ with high probability under the uniform permutation distribution} if, over uniformly random permutation $\sigma$ and for any $d \geq 1$, routing $r D_\sigma$ on a random $\mathcal{S} \sim \mathscr{S}$ uses paths of maximum latency $L$ with probability at least $1-C_d/N^d$, where $C_d$ is a constant that may depend on $d$. In the special case when $\mathscr{S}$ is a point-mass distribution on a singleton set $\{\mathcal{S}\}$, we say that the fixed design $\mathcal{S}$ achieves maximum latency $L$ with high probability under the uniform distribution over permutation demands.}
\end{dfn}

\begin{dfn} \label{def:round-robin}
	A {\em round robin} for a group of nodes $S$ of size $k$, $\{s_0,\hdots,s_{k-1}\}$ is a schedule of $k-1$ timesteps in which each element of $S$ has a chance to send directly to each other element exactly once; during timestep $t \in [k-1]$ node $s_i$ may send to $s_{i+t\mod{k}}$.
\end{dfn}

%% file: upper-bound.tex
\section{Upper Bound: Oblivious Design} \label{sec:upper-bound}

In this section we prove \Cref{thm:tradeoff}, parts~\ref{to.orn.whp} and~\ref{to.orn.unif},
restated below.\\

\noindent \textbf{\Cref{thm:tradeoff}.\ref{to.orn.whp}-\ref{thm:tradeoff}.\ref{to.orn.unif}.} \textit{Given any fixed throughput value $r\in(0,\frac{1}{2}]$, let $\semih = \semih(r) =\lfloor\frac{1}{r}-1\rfloor$,
and let $L_{upp}(r,N)$ be the function}
\begin{equation*}
	L_{upp}(r,N) = \semih N^{1/\semih} 
\end{equation*}
\textit{Then assuming $\frac1r \not\in \mathbb{Z}$, 
  there exists a family of distributions over ORN designs for infinitely many network sizes $N$ which attains maximum latency $\bigotilde(L_{upp}(r,N))$, and achieves throughput $r$ with high probability.}
\textit{Furthermore, under the same assumption on $\eps$, for infinitely many network sizes there exists a fixed distribution over ORN designs which attains maximum latency $\bigotilde(L_{upp}(r,N))$, and achieves throughput $r$ with high probability under the uniform distribution.}
\vspace{3mm}

We will begin by constructing an ORN design $\mathcal{R}^0$ which is parameterized by $N$, $\semih$, and $C$, where $C$ is a parameter which we set during our analysis to a suitable function of $N$ and $r$ designed to achieve the appropriate tradeoffs between throughput and latency.
We will then analyze $\mathscr{R}_N(\semih,C)$, a distribution over all ORN designs $\mathcal{R}^\tau$ which are equivalent to $\mathcal{R}^0$ up to re-labeling of nodes, and show that it satisfies the conclusion of \Cref{thm:tradeoff}.\ref{to.orn.whp}. Furthermore we will show that the fixed design $\mathcal{R}^0$ itself satisfies the conclusion of \Cref{thm:tradeoff}.\ref{to.orn.unif}.

\subsection{Connection Schedule} \label{sec:orn-conn-sched}
The connection schedule of $\mathcal{R}^0$, like the Vandermonde Basis Scheme of \cite{stoc-paper}, is based on round-robin phases (cf.\ \Cref{def:round-robin}) defined by Vandermonde vectors.
We interpret the set of nodes as elements of the vector space $\mathbb{F}_p^\semih$ over the prime field $\mathbb{F}_p$, where $N=p^\semih$. 
Each node $a\in[N]$ can then be interpreted as a unique $\semih$-tuple $(a_1,a_2,\hdots,a_{\semih})\in\mathbb{F}_p^\semih$. 

During this connection schedule, each node will participate in a series of round robins, each defined by a single Vandermonde vector of the form $\bm{v}(x)=(1,x,x^2,\hdots,x^{\semih-1})$. 
The period length of the connection schedule is
  $T = C(g+1)(p-1)$, and one full period of the schedule consists
  of $C(g+1)$ consecutive round robins called \emph{Vandermonde
    phases} or simply \emph{phases},
  each of length $(p-1)$ timesteps.
  The $C(g+1)$ phases constituting one
  period of the schedule are defined by distinct Vandermonde
  vectors of the form $\bm{v}(x) = (1,x,\ldots,x^{\semih-1})$.
  No property of the Vandermonde vectors other than distinctness
  is required. Since Vandermonde vectors are parameterized by
  elements $x \in \mathbb{F}_p$, we require $p \geq C(g+1)$
  to ensure that sufficiently many distinct Vandermonde vectors
  exist. The set of Vandermonde phases in one period of the
  schedule will be grouped into $(\semih+1)$ non-overlapping {\em phase blocks}, each phase block consisting of $C$ phases.

  More formally, we identify each congruence class
    $k \pmod{T}$ with 
    a phase number $x$ and a scale factor $s$, $0\leq x <p$ and $1\leq s<p$, such that $k = (p-1) x + s - 1$. It is useful to think of timesteps as being indexed by ordered pairs $(x,s)$ rather than by the corresponding congruence class mod $T$, so we will sometimes abuse notation and refer to timestep $(x,s)$ in the sequel, when we mean $k = (p-1) x + s - 1$. The connection schedule of $\mathcal{R}^0$, during timesteps $t \equiv k \pmod{T}$, uses permutation $\pi^0_k(a) = a + s \bm{v}(x)$, where $x$ and $s$ are the phase number and scale associated to $k$. Thus, 
    each phase takes $(p-1)$ timesteps, and allows each node $a$ to connect with nodes $a'$ where the difference $a'-a$ belongs to the one-dimensional linear subspace generated by $\bm{v}(x)$.

As described above, $\mathscr{R}_N(\semih,C)$ is a distribution over all ORN designs $\mathcal{R}^\tau$ which are equivalent to $\mathcal{R}^0$ up to re-labeling.
When we sample a random design $\mathcal{R}^\tau$, we sample a uniformly random permutation of the node set $\tau:\mathbb{F}_p^h\rightarrow \mathbb{F}_p^\semih$, producing the schedule $\pi_{k}^\tau(a) = \tau^{-1}\Big(\pi_{k}^0\big(\tau(a) \big) \Big)$.
Note that, for every edge from node $a$ to node $\pi_t^{\tau}(a)$ in $\mathcal{R}^{\tau}$, there is a unique equivalent edge from $\tau(a)$ to $\tau(\pi_t^{\tau}(a))$ in $\mathcal{R}^0$.

\subsection{Routing Scheme} \label{sec:orn-rout-scheme}

Our routing scheme for $\mathcal{R}^0$
constructs routing paths composed
of at most one physical hop in each of $\semih+1$
consecutive phase blocks. Such a path can be identified
by the node and timestep at which it originates, the
phases in which it traverses a physical hop, and the
scale factors applied to the Vandermonde vectors
defining each of those phases. Our first definition
specifies a structure called a {\em pseudo-path}
that encodes all of this information.
\begin{dfn} \label{dfn:pseudo-path}
  A $k$-hop {\em pseudo-path} from $a$ to $b$ starting at time $t$
  is a sequence of ordered pairs
  $(x_1,\alpha_1), \ldots, (x_k,\alpha_k)$ such that:
  \begin{itemize}
  \item $x_1,\ldots,x_{k}$ are phases belonging to distinct, consecutive phase blocks beginning with the first complete phase block
    after time $t$;
  \item $\alpha_1,\ldots,\alpha_{k} \in \mathbb{F}_p$ are scalars;
  \item $b-a = \alpha_1 \bm{v}(x_1) + \alpha_2 \bm{v}(x_2) + \cdots +
  \alpha_{k} \bm{v}(x_{k})$.
  \end{itemize}
  A \emph{non-degenerate} pseudo-path is one satisfying
  $\alpha_1 \neq 0$ and $\alpha_k \neq 0$. 
  
  The path corresponding to a pseudo-path is the path in the
  virtual topology that starts at $a$,
  traverses physical edges in timesteps
  $k_i = (x_i,\alpha_i)$ for all $i$ such that $\alpha_i \neq 0$,
  and traverses virtual edges in all other timesteps.
\end{dfn}
Note that the path corresponding to a $k$-hop
  pseudo-path may contain fewer than $k$ physical hops.
  Two distinct pseudo-paths may correspond to
  the same path, if the only difference between the pseudo-paths
  lies in the timing of the phases with $\alpha_j=0$, i.e.~the
  phases in which no physical hop is taken.
  Distinguishing between pseudo-paths
  that correspond to the same path is unnecessary for the purpose 
  of describing the edge sets of routing paths, but it turns out
  to be 
  essential for the purpose of defining and analyzing
  the \emph{distribution} over
  routing paths employed by our routing schemes.

  Our oblivious routing scheme for $\mathcal{R}^0$
  divides flow among routing paths in proportion to
  a probability distribution over paths defined as follows. 
  To sample routing path from $a$ to $b$ starting at
  time $t$, we sample a uniformly random 
  non-degenerate $(\semih+1)$-hop pseudo-path
  from $a$ to $b$ that starts
  at time $t$.
  We then translate this pseudo-path into the corresponding
  path, and use that as a routing path from $a$ to $b$.
  In other words, our oblivious routing scheme divides
  flow among paths in proportion to the number
  of corresponding non-degenerate $(\semih+1)$-hop pseudo-paths.

  To analyze the oblivious routing scheme, or even
  to confirm that it is well-defined, it
  will help to prove a lower bound on the number
  of solutions to the equation
  \begin{equation} \label{eq:vdm}
  b-a = \alpha_1 \bm{v}(x_1)
  + \cdots + \alpha_{\semih+1} \bm{v}(x_{\semih+1})
  \end{equation}
  that satisfy $\alpha_1 \neq 0, \,
  \alpha_{\semih+1} \neq 0$. For any
  $i \in [g+1]$ and $\beta \in \mathbb{F}_p$,
  there is a unique solution to~\eqref{eq:vdm}
  with $\alpha_i=\beta$. This is because
  the equation
  \[
  b-a - \beta \bm{v}(x_i) =
  \sum_{j \neq i} \alpha_j \bm{v}(x_j)
  \]
  is a system of $g$ linear equations
  in $g$ unknowns, with an invertible
  coefficient matrix. (Here we have used
  the fact that the vectors $\bm{v}(x_j)$
  are distinct Vandermonde vectors, hence
  linearly independent.) Hence, the
  total number of solutions of~\eqref{eq:vdm}
  is $p$, and there is exactly one solution
  with $\alpha_1=0$ and exactly one solution
  with $\alpha_{\semih+1}=0$. The
  number of solutions with
  $\alpha_i \neq 0$ and $\alpha_{\semih+1} \neq 0$
  is therefore either $p-2$ or $p-1$.
  Since there are $C^{\semih+1}$ ways to choose
  the $\semih+1$ distinct phases $x_1,\ldots,x_{\semih+1}$,
  we conclude that the number of
  non-degenerate $(g+1)$-hop pseudo-paths
  from $a$ to $b$ starting at time $t$ is
  between $(p-2) C^{g+1}$ and $(p-1) C^{g+1}$.

The routing scheme of $\mathcal{R}^\tau$, for general $\tau$, is defined using the bijection between the edges of $\mathcal{R}^\tau$ and those of $\mathcal{R}^0$. 
For any path from node $a$ to node $b$ in $\mathcal{R}^{\tau}$ there is a unique equivalent path from $\tau(a)$ to $\tau(b)$ in $\mathcal{R}^0$.
To route from $a$ to $b$ in $\mathcal{R}^\tau$, simply apply the inverse of this bijection to the probability distribution over routing paths from $\tau(a)$ to $\tau(b)$ in $\mathcal{R}^0$.

\subsection{Latency-Throughput Tradeoff} \label{sec:ub-tradeoff}

It is clear that any design $\mathcal{R}^\tau \sim\mathscr{R}_N(\semih,C)$ will have maximum latency $C(\semih+2)(p-1) < C(\semih+2)N^{1/\semih}$. (The factor of $\semih+2$ reflects the fact that messages wait for the duration of at most one phase block, then use the following $\semih+1$ phase blocks to reach their destination.)
Thus, we focus on proving the achieved throughput rate with high probability in this section. Parts~\ref{to-sec3.orn.unif} and~\ref{to-sec3.orn.whp} of the following theorem correspond to parts~\ref{to.orn.unif} and~\ref{to.orn.whp} of Theorem~\ref{thm:tradeoff}, respectively.

\begin{restatable}{theorem}{thmfixeddesignthroughput}\label{thm:fixed-design-tradeoff} \label{thm:orn-tradeoff} 
  Given a fixed throughput value $r$, let $\semih = \semih(r) = \lfloor\frac{1}{r}-1\rfloor$ and $\eps = \eps(r) = \semih + 1 - (\frac{1}{r}-1)$, and assume $\eps\neq 1$.
  As $N$ ranges over the set of prime powers $p^{\semih}$
for primes $p$ exceeding $\max \left\{
C(\semih+1), 2 + \frac{2}{1-\eps} \right\}$,
let $\gamma=\ln \left( \frac{\semih - \eps - 2/(p-2)}{\semih-1} \right)$ and $C=\frac{\log\log N}{\gamma^2}\ln(N)$.
Then:
\begin{enumerate}
\item \label{to-sec3.orn.unif}
  the design $\mathcal{R}^0$ achieves throughput
  $r$ with high probability under the uniform distribution, 
\item \label{to-sec3.orn.whp}
  the family of distributions $\mathcal{R}_N(\semih,C)$
  achieves throughput $r$ with high probability.
\end{enumerate}
\end{restatable}

Note that if $\eps=1$, i.e.~if $\frac1r \in \mathbb{Z}$, then there are no primes $p$ which exceed $2 + \frac{2}{1-\eps}$, therefore we condition against $\eps=1$.

Both parts of the theorem will be proven by focusing on
  the congestion of physical edges in the design $\mathcal{R}^0$.
  For the first part, the focus on edges in $\mathcal{R}^0$ is
  obvious. For the second part, we make use of 
  the isomorphism between $\mathcal{R}^{\tau}$ and
  $\mathcal{R}^0$. Rather than considering a fixed demand function
  $D$ and random design $\mathcal{R}^{\tau}$, we may consider a
  fixed design $\mathcal{R}^0$ and 
  random demand function $D^{\tau}(t) = P^{-1} D(t) P$
  where $P$ denotes the permutation matrix with
  $P_{i,\tau(i)} = 1$ for all $i$.

Now, focusing on any particular
  edge $e \in \evirt(\mathcal{R}^0)$, we bound the probability
  that $e$ is overloaded by breaking down the (random) amount
  of flow traversing $e$ as a sum, over $0 \leq q \leq \semih$,
  of the amount of flow that crosses $e$ on the $(q+1)$-th hop
  of a routing path. We will describe how to interpret each
  of these random amounts of flow as the value of a bilinear
  form on a pair of vectors randomly sampled from an orbit of
  a permutation group action. (The bilinear form is
  related to the demand function $D$, and the pair of
  vectors is related to the routing scheme.) We will
  then use a Chernoff-type bound for the values of bilinear forms
  on permutation group orbits, to bound the probability that
  the amount of $(q+1)$-th hop flow crossing $e$ is larger
  than average.
  Finally we will impose a union bound to show the probability
  that any edge gets overloaded is extremely small.

Existing Chernoff-type bounds for
negatively associated random variables
are sufficient for the tail bound in the
first part of the theorem, but not
for the second part. (See \Cref{rmk:tailbound-double-isnt-NA}
below.) Instead, we prove the following novel tail bound for
the distribution of bilinear sums on orbits of a
permutation group action.

\begin{restatable}{theorem}{thmtailbound} \label{thm:tailbound}
  Suppose $\vc{u}, \, \vc{v} \in (\reals_{\geq 0})^N$ are
  non-zero, non-negative vectors satisfying
  \begin{equation} \label{eq:tailbound-hyp}
    \left( \frac{\|\vc{u}\|_1}{\|\vc{u}\|_{\infty}} \right)
    \left( \frac{\|\vc{v}\|_1}{\|\vc{v}\|_{\infty}} \right)
    \; \geq \; C N
  \end{equation}
  for some $C \geq 1$. Let $D$ be any $N$-by-$N$ doubly
  stochastic matrix and consider the bilinear form 
    \begin{equation} \label{eq:bilinear-form}
      B(\vc{x},\vc{y}) = \sum_{i \neq j} D_{ij} x_i y_j .
    \end{equation}
    Let $M = 1$ if $D$ is a permutation matrix, and $M = N^2$ otherwise.
    If $P$ is a uniformly random $N$-by-$N$ permutation matrix then:
  \begin{enumerate}
  \item for any $\gamma > 0$,
    \begin{equation} \label{eq:tailbound-single}
      \Pr \left( B(\vc{u},P \vc{v})  \, \geq \, e^{\gamma} \frac{\|\vc{u} \|_1 \, \|\vc{v}\|_1}{N} \right)
      \leq M e^{-\frac12 \gamma^2 C} ;
     \end{equation}
  \item for any $\gamma > 0$,
    \begin{equation} \label{eq:tailbound-double}
      \Pr \left( B(P \vc{u}, P \vc{v}) \, \geq \, e^{\gamma} \frac{\|\vc{u} \|_1 \, \|\vc{v}\|_1}{N} \right)
      \leq 15 M e^{-\frac{1}{100} \gamma^2 C} .
    \end{equation}
  \end{enumerate}
\end{restatable}

The proof of \Cref{thm:tailbound} is deferred to
\Cref{sec:tailbound}.

\subsection*{Proof of \Cref{thm:fixed-design-tradeoff}}\label{sec:ub-proof}

\begin{proof}
We may assume without loss of generality that
  the demand matrix $D(t)$ is doubly stochastic
  for all $t$. For part 1 of the theorem this is
  because $D(t)$ is assumed to be a random
  permutation matrix. For part 2, it is because
  every non-negative matrix whose row and column
  sums are bounded above by 1 can be made into
  a doubly stochastic matrix by (weakly) increasing
  each of the matrix entries~\cite{stoc-paper}.
  Modifying the demand function in this way
  cannot decrease the induced flow on any edge,
  so it cannot increase the probability that
  $f(R,rD)$ is feasible. Thus, we will assume for
  the remainder of the proof that $D(t)$ is doubly
  stochastic for all $t$.
  
Fix an edge $e$ and $0\leq q\leq \semih$, and consider the amount of flow traversing edge $e$ traveling on paths where edge $e$ occurs in the $(q+1)$-th phase block\footnote{We number phase blocks in a flow path using the convention that phase block 1 is the first {\em complete} phase block in the flow path. Recall from \Cref{sec:orn-rout-scheme} that this is also the first phase block in which it is possible that the flow is transmitted on a physical edge.} of the flow path. 
We will denote this value as the {\em amount of $(q+1)$-th hop flow traversing edge $e$}.\footnote{Note this is a different value than if edge $e$ is the $(q+1)$-th physical hop traversed on the path. It may be the case that in some earlier phase blocks of the path, flow may not have traversed any physical hop. If this is confusing, revisit {\em pseudo-paths} in \Cref{sec:orn-rout-scheme}.}

First we examine $q=0$.
First-hop flow traversing edge $e$ originates at
  source node $\tl(e)$ during the
  phase block preceding the one to which $e$ belongs.
  There are $C(p-1)$ time steps during that phase block,
  and $r$ units of flow per time step originate at
  $\tl(e)$. Each unit of flow is divided evenly among
  a set of at least $(p-2) C^{\semih+1}$ pseudo-paths, at
  most $C^{\semih}$ of which begin with edge $e$ as their
  first hop. (After fixing the first hop and the destination
  of a $(\semih+1)$-hop pseudo-path, the rest of the path
  is uniquely determined by the $g$-tuple of phases
  $x_2,\ldots,x_{g+1}$.) Hence, of the $r C (p-1)$ units
  of flow that could traverse $e$ as their first hop,
  the fraction that actually do traverse $e$ as their
  first hop is at most $\frac{C^{\semih}}{(p-2) C^{\semih+1}}$.
  Consequently, the amount of first-hop flow on $e$ is
  bounded above by
  $\frac{r C (p-1) \cdot C^{\semih}}{(p-2) C^{\semih+1}}
  = \left( \frac{p-1}{p-2} \right) r .$
  (Note that this is not a probabilistic statement;
  the upper bound on first-hop flow holds with probability 1.)
  A symmetric argument shows that the amount of
  last-hop flow on $e$ is bounded above by
  $\left( \frac{p-1}{p-2} \right) r $
  as well.

Now suppose $1\leq q\leq \semih-1$, and let $X_i$ be the random variable realizing the amount of $(q+1)$-th hop flow traversing edge $e$ due to source node $i$. Clearly, the total amount of $(q+1)$-th hop flow traversing $e$ will be $\sum_i X_i$.
  Let $I$ denote the interval of timesteps
  constituting the $q^{\mathrm{th}}$ phase block
  before the phase block that contains edge $e$;
  recall that this means $I$ is made up of $C(p-1)$
  consecutive timesteps.
  Let 
  \[
  \overline{D}_{ij} = \frac{1}{r C (p-1)} \sum_{t \in I} D(t)_{ij}
  \]
  denote the (normalized) rate of flow demanded by
  source-destination pair $(i,j)$ during phase block $I$.
  The normalizing factor makes $\overline{D}$ into a
  doubly stochastic matrix.
Let $\rho^-_q(i,e)$ denote the number of $q$-hop pseudo-paths from $i$ to $\tl(e)$ with non-zero first coefficient, and let $\rho^+_{\semih-q}(e,j)$ denote the number of $(\semih-q)$-hop pseudo-paths from $\hd(e)$ to $j$ with non-zero last coefficient.
Finally, let $\rho_{\semih+1}(i,j)$ denote the number of
  non-degenerate $(\semih+1)$-hop pseudo-paths from
  $i$ to $j$. Of the flow that originates at $i$ with
  destination $j$ during time window $I$, the fraction
  of flow that traverses edge $e$ under our routing scheme
for $\mathcal{R}^0$ is $\rho^-_q(i,e) \cdot \rho^+_{\semih-q}(e,j) / \rho_{\semih+1}(i,j)$. Hence,
\begin{align}
  \nonumber
  X_i & = \sum_{j \in [N], \, j \neq i} \frac{\rho^-_q(i,e) \cdot \rho^+_{\semih-q}(e,j)}{\rho_{\semih+1}(i,j)} \cdot \left( \sum_{t \in I} D(t)_{ij} \right) \\
  \nonumber
      & \leq 
  \sum_{j \in [N], \, j \neq i} \frac{\rho^-_q(i,e) \cdot \rho^+_{\semih-q}(e,j) \cdot r C (p-1) \cdot \overline{D}_{ij}}{(p-2) C^{\semih+1}} \\
  \nonumber
      & = \left( \frac{p-1}{p-2} \right) r
  \sum_{j \in [N], \, j \neq i} \overline{D}_{ij} \left( \frac{\rho^-_q(i,e)}{C^q} \right) \left( \frac{\rho^+_{\semih-q}(e,j)}{C^{\semih-q}} \right) \\
  \label{eq:ub.sumx}
      \sum_{i=1}^{N} X_i & \leq \left( \frac{p-1}{p-2} \right) r 
      \sum_{i \neq j} \overline{D}_{ij} \left( \frac{\rho^-_q(i,e)}{C^q} \right) \left( \frac{\rho^+_{\semih-q}(e,j)}{C^{\semih-q}} \right) =
      \sum_{i \neq j} \overline{D}_{ij} u_i v_j
\end{align}
where
\begin{equation} \label{eq:ub.uv}
  u_i = \left( \frac{p-1}{p-2} \right) r \left(
  \frac{\rho^-_q(i,e)}{C^q} \right), \qquad 
  v_j = \frac{\rho^+_{\semih-q}(e,j)}{C^{\semih-q}} .
\end{equation}

To prove the first part of the theorem, \Cref{thm:orn-tradeoff}.1, when the
ORN design is fixed to be $\mathcal{R}^0$ and
the demand function is the time-stationary
demand $D_{\sigma}$ for a random permutation
$\sigma$, then
\[
\sum_{i \neq j} \overline{D}_{ij} u_i v_j =
\sum_{i \neq \sigma(i)} u_i v_{\sigma(i)}  \leq
\sum_{i=1}^N u_i v_{\sigma(i)} .
\]
The distribution of $\sigma$ is the same as
the distribution of $\tau \circ \pi$ where
$\pi$ is an arbitrary (non-random) permutation
without fixed points, and $\tau$ is a uniformly random
permutation. Letting $P$ denote the permutation matrix
representing $\tau$, the amount of $(q+1)^{\mathrm{th}}$
hop flow on edge $e$ is stochastically dominated
by 
\[
\sum_{i=1}^N u_i v_{\tau(\pi(i))} =
B_{\pi}(\vc{u}, P \vc{v})
\]
where $B_{\pi}$ denotes the bilinear
form $B_{\pi}(\vc{x},\vc{y}) = \sum_{i=1}^N x_i y_{\pi(i)}$.

Similarly, to prove the second part of the theorem, \Cref{thm:orn-tradeoff}.2,
recall that we are drawing a random ORN
design $\mathcal{R}^{\tau}$ from the distribution
$\mathscr{R}_N(C,r)$, and that the induced
$(q+1)$-th hop flow
on the edge of $\mathcal{R}^{\tau}$ corresponding
to $e$, under demand function $D$, is equal to
the induced $(q+1)$-th hop flow on edge $e$ under demand function
$P^{-1} D P$. Again letting $P$ denote the permutation
matrix representing $\tau$, this induced flow is
bounded above by
\[
\sum_{i \neq j} (P^{-1} \overline{D} P)_{ij} u_i v_j
= \sum_{i \neq j} \overline{D}_{ij} u_{\tau(i)} v_{\tau(j)}
= B(P \vc{u}, \, P \vc{v} )
\]
where $B$ is the bilinear form
$B(\vc{x}, \vc{y}) = \sum_{i \neq j} \overline{D}_{ij} x_i y_j$.

Hence, we are in a position to prove tail bounds
on the induced $(q+1)$-th hop flow on edge $e$,
using the Chernoff-type bounds in \Cref{thm:tailbound},
provided we can estimate the norms $\| \vc{u} \|_1, \,
\| \vc{v} \|_1, \, \| \vc{u} \|_{\infty}, \,  \| \vc{v} \|_{\infty}.$
For $\| \vc{u} \|_1$ we have
$
  \| \vc{u} \|_1 = \frac{p-1}{p-2} \cdot \frac{r}{C^q} \cdot
  \sum_{i=1}^N \rho^-_q(i,e). $
The sum on the right side can be calculated by realizing
that it counts the total number of $q$-hop pseudo-paths
with non-zero first coefficient that end at $\tl(e)$.
There are $C^q$ ways of choosing a $q$-tuple of phases
from the $q$ phase blocks preceding the phase block
containing $e$, for each such choice there are
$(p-1)p^{q-1}$ ways to choose a sequence of coefficients
beginning with a non-zero value. Hence,

\begin{equation*} 
  \| \vc{u} \|_1 = \frac{p-1}{p-2} \cdot \frac{r}{C^q} \cdot
  (p-1) p^{q-1} C^q = \frac{(p-1)^2}{p(p-2)} \cdot p^q \cdot r .
\end{equation*}
Similarly,
\begin{equation*} 
  \| \vc{v} \|_1 = \frac{p-1}{p} \cdot p^{\semih-q} .
\end{equation*}

Now we turn to bounding $\|\vc{u}\|_{\infty}, \,
\| \vc{v} \|_{\infty}$ from above, which is tantamount
to bounding the number of $q$-hop pseudo-paths
from $i$ to $\tl(e)$ and $(\semih-q)$-hop
pseudo-paths from $\hd(e)$ to $j$, with
non-zero first and last coefficients respectively.
One such upper bound is easy to derive:
for each of the $C^q$ many ways of selecting
one phase $\vc{x}_i$ from each of the $q$ phase
blocks preceding $\tl(e)$, there is at most
one $q$-hop pseudo-path from $i$ to $\tl(e)$
using that sequence of phases. This is because
the existence of two distinct such pseudo-paths
would imply that the vector $\tl(e) - i$ could be represented in
two distinct ways as a linear combination of
vectors in the set $\{\vc{x}_1, \ldots, \vc{x}_q\}$,
violating linear independence. For an analogous reason,
$\rho^+_q(\hd(e),j) \leq C^{\semih-q}$.

However, if $q \leq \semih/2$ then there is a
tighter upper bound: 
$\rho^-_q(i,\tl(e)) \leq C^{q-1}$.
To see why, first observe that any $2q$
of the $C(\semih+1)$ Vandermonde vectors
used in the $\semih+1$ phase blocks preceding
edge $e$ must be linearly independent,
since $2q \leq \semih$. If
$(x_1,\alpha_1),\ldots,(x_q,\alpha_q)$
and $(x'_1,\alpha'_1),\ldots,(x'_q,\alpha'_q)$
are two pseudo-paths from $i$ to
$\tl(e)$ then 
\[
\{(x_i,\alpha_i) \mid \alpha_i \neq 0\} \; = \;
\{(x'_j,\alpha'_j) \mid \alpha'_j \neq 0 \},
\]
as otherwise the vector $(\tl(e) - i)$ could
be represented in two inequivalent ways as a
linear combination of elements of
$\{x_1,x'_1,x_2,x'_2,\ldots,x_q,x'_q\}$,
contradicting linear independence.
Consequently, when $q \leq g/2$, two
distinct $q$-hop pseudo-paths from
$i$ to $\tl(e)$ can only differ in
the choice of phases $x_i$ with $\alpha_i = 0$.
In other words, every $q$-hop pseudo-path
from $i$ to $\tl(e)$ has the same coefficient
sequence $\alpha_1,\alpha_2, \ldots, \alpha_q$,
and in constructing the corresponding
phase sequence we have only one choice of
phase when $\alpha_i \neq 0$ and
$C$ choices when $\alpha_i = 0$.
Furthermore, there is at least one value
of $i$, namely $i=1$, for which $\alpha_i \neq 0$.
Consequently, $\rho^-_q(i,\tl(e)) \leq C^{q-1}$
when $q \leq \semih/2$, as claimed. An analogous
argument proves that $\rho^+_q(\hd(e),j) \leq
C^{\semih-q-1}$ when $\semih-q \leq \semih/2$. For every
$q$, at least one of $q, \semih-q$ is less
than or equal to $\semih/2$, and hence
\begin{align*}
  \rho^-_q(i,\tl(e)) \cdot
  \rho^+_q(\hd(e),j) & \leq
  \max \{ C^{q-1} \cdot C^{\semih-q}, \,
  C^q \cdot C^{\semih-q-1} \} = C^{\semih-1} \\
  \| \vc{u} \|_{\infty} \| \vc{v} \|_{\infty}
  & \leq
  \left( \frac{p-1}{p-2} \right) r
  \left( \frac{\rho^-_q(i,\tl(e)) \cdot
    \rho^+_q(\hd(e),j)}{C^\semih} \right) 
  \leq
  \left( \frac{p-1}{p-2} \right) \frac{r}{C} \\
  \left( \frac{ \| \vc{u} \|_1 \| \vc{v} \|_1 }{ \| \vc{u} \|_{\infty} \| \vc{v} \|_{\infty} } \right) & \geq
  \frac{ \frac{ (p-1)^3 }{p^2 (p-2)} \cdot p^{\semih} \cdot r }
  { \frac{p-1}{p-2} \cdot \frac{r}{C} } 
   = \left( \frac{p-1}{p} \right)^2 C N
   \geq \frac12 C N
\end{align*}
for $p \geq 5$. If we observe that
$\frac{ \| \vc{u} \|_1 \| \vc{v} \|_1 }{N}
= \frac{(p-1)^3}{p^2 (p-2)} r <
r,$ then
we may use \Cref{thm:tailbound} to conclude
that for any $\gamma > 0$,
\begin{align*}
  \Pr \left( B_{\pi}(\vc{u}, P \vc{v} ) \geq 
   e^{\gamma} r \right)
  & \leq \; N^2 e^{- \frac14 \gamma^2 C} \\
  \Pr \left( B(P \vc{u}, P \vc{v} ) \geq 
   e^{\gamma} r \right)
  & \leq 15 N^2 e^{- \frac1{200} \gamma^2 C } .
\end{align*}
  Supposing $C \geq \frac{\log\log N}{\gamma^2}\ln(N)$ for some positive integer, then we union bound over all $C (p-1) (g+1) N$  edges of the virtual topology and all $1 \leq q \leq g-1$ to find
\begin{align*}
	\Pr[&\text{any edge has}\geq e^{\gamma} r \text{ } (q+1)\text{-th hop flow for some }1\leq q\leq \semih-1]\\ 
	&\leq NC(p-1)(g+1)(g-1) \cdot 15 N^2 \left( e^{-\frac{1}{200}\gamma^2} \right)^{C} \\
	&\leq N^{3+1/\semih} \frac{\log\log N}{\gamma^2}\ln(N) (\semih^2-1) e ^{-\frac{1}{200} \log\log N\ln(N)} \\
	&\leq \left( N^{3+1/\semih}\frac{\log\log N\ln(N)}{\gamma^2} (\semih^2-1) \right)N^{- \frac{1}{200} \log\log N} \\
	&\leq \bigo\left(\frac{1}{\gamma^2 N^d}\right) \text{ for any constant }d.
\end{align*}
This fulfills our definition of with high probability for fixed $\gamma$.

Finally, we need to show that if none of the bad events as described above occur, if every edge has at most $e^{\gamma} r$ $(q+1)$-th hop flow for $1 \leq q \leq g-1$, then no edge will be overloaded.
Recall also that the $(q+1)$-th hop flow on $e$ for $q \in \{0,g\}$ is $\left( \frac{p-1}{p-2} \right) r = r + \frac{r}{p-2}$. 
Recall also that $e^{\gamma} = \frac{\semih - \eps - 2/(p-2)}{\semih-1}$, $\semih = \lfloor\frac{1}{r}-1\rfloor$, and $\eps = \semih + 1 - \left( \frac{1}{r}-1 \right) = 2 + \semih - \frac1r$. Hence, if no bad events occur, the induced flow
on each edge will be bounded above by
\begin{align*}
  2 r + \frac{2r }{p-2} + (\semih-1) e^{\gamma} r & = \left(2 + \frac{2}{p-2} + \semih - \eps - \frac{2}{p-2} \right) r 
	& = \left( 2 + \semih-\eps \right) r = \left( \frac1r \right) r = 1 . 
\end{align*}
\end{proof}

%% file: tailbound.tex
\subsection{A Tail Bound for Bilinear Sums}
\label{sec:tailbound}

In \Cref{sec:ub-tradeoff}, our analysis of the distribution of
the amount of flow traversing an edge $e$ depends on certain
tail bounds for the distribution of bilinear sums on orbits
of a permutation
group action.
The relevant tail bound is stated as
\Cref{thm:tailbound} above. 
This section is devoted to proving the theorem.
The proof will make use of a Chernoff-type concentration
bound for negatively associated random variables. We begin by recalling
some definitions and facts about negative association; see
\cite{dubhashi1996balls,joag1983negative,neg-assoc-notes} for
an introduction to this topic.

\begin{dfn}[\cite{joag1983negative,khursheed1981positive}] \label{def:negative-association}
A set of random variables $X_1,...,X_n$ are {\em negatively associated} if for any two functions $f,g:\mathbb{R}^n\rightarrow \reals$ that are either both monotone increasing or both monotone decreasing, and dependent on\footnote{For the purposes of this definition, an $n$-variate function $f$ is dependent on a set of indices $I\subseteq[n]$ if $f(x_1,\hdots,x_n) = f(y_1,\hdots,y_n)$ holds whenever $x_i=y_i$ for all $i\in I$.} disjoint subsets of indices $S_f,S_g\subseteq [n]$, then 
\[ \E[f(\xvec) \cdot g(\xvec)] \leq \E[f(\xvec)] \cdot \E[g(\xvec)] \]
\end{dfn}

Many examples of negatively associated random variables can be constructed using
the following definition and lemma.
\begin{dfn} \label{def:const-ordered-rows}
An $n$ by $m$ matrix $A$ has {\em consistently ordered rows} if there exists some permutation $\pi: [m]\rightarrow[m]$ of the columns of $A$ such that for all rows $i\in[n]$, $A[i,\pi(1)] \leq \cdots \leq A[i,\pi(m)]$. 
\end{dfn}
\begin{restatable}{lemma}{lemconstorderedrows}\label{lem:constorderedrows}
Suppose $A$ is an $n$ by $n$ matrix, and $X_1,...,X_n$ are random variables sampled by the following process: sample a permutation $\pi:[n]\rightarrow [n]$ uniformly at random, and set $X_i = A[i,\pi(i)]$.
If the entries of $A$ are non-negative and $A$ has consistently ordered rows, then $X_1,...,X_n$ are negatively associated.
\end{restatable}
\begin{proof} 
	This will be proved by induction on $n$. Note that negative association amounts to showing that the covariance $Cov(f(\xvec), g(\xvec)) \leq 0$. WLOG, since $A$ has consistently ordered rows, we can assume that $A[i,1]\leq...,A[i,n]$ for all $i\in[n]$.
	
	Base case: $n=2$. Then $A$ is a 2 by 2 matrix, and since $f$,$g$ are both either monotone increasing or monotone decreasing, then
	\begin{align*}
		Cov(f(\xvec), g(\xvec)) = & \mbox{ } \frac{1}{4}\Big( f(A[1,1])g(A[2,1]) + f(A[1,2])g(A[2,2]) \\
		& - f(A[1,1])g(A[2,2]) - f(A[1,2])g(A[2,1]) \Big) \leq 0
	\end{align*}
	
	Now suppose the lemma is true for $n=k$, and for now suppose $f,g$ are both monotone increasing. We will need two properties of covariance.
	
	\textit{Property 1: (law of total covariance)} Let $X,Y,$ and $Z$ be any random variables. Then $Cov(X,Y) = \E[Cov(X,Y)|Z] + Cov(\E[X|Z],\E[Y|Z])$.
	
	\textit{Property 2: (Chebyshev's algebraic inequality)} Given a random variable $Z$ and monotone increasing $h_1$ and monotone decreasing $h_2$, then $Cov(h_1(Z),h_2(Z)) \leq 0$.
	
	Now, consider the random variable $I = \pi^{-1}(1)$. This indicates which random variable $X_i$ realizes its smallest value. Then by Property 1, 
	\[Cov(f(\xvec), g(\xvec)) = \E\big[ Cov(f(\xvec), g(\xvec))|I \big] + Cov\big( \E\big[ f(\xvec)|I \big], \E\big[ g(\xvec)|I \big] \big)\]
	
	For any fixed $I$, the first term is random over 1 fewer variable, meaning this falls under the inductive hypothesis and is $\leq 0$.
	
	To show the second term is $\leq 0$, we will show that as functions of $I$, one of $\E\big[ f(\xvec)|I \big]$ or $\E\big[ g(\xvec)|I \big]$ is monotone increasing, and the other is monotone decreasing.
	
	Due to how the random variables $X_i$ are chosen from $A$, they can be equivalently chosen from any matrix $A'$ equivalent up to a re-ordering of rows. We will re-order the rows of $A$ to enforce $h_1(I) = \E\big[ f(\xvec)|I \big]$ monotone increasing and $h_2(I) = \E\big[ g(\xvec)|I \big]$ monotone decreasing in $I$.
	
	Let $\sigma_f:[|S_f|]\rightarrow S_f$ impose the ordering $h_1(\sigma_f(1)) \leq \hdots \leq h_1(\sigma_f(|S_f|))$. Additionally, let $\sigma_g:[|S_g|]\rightarrow S_g$ impose $h_2(\sigma_g(1))\leq \hdots \leq h_2(\sigma_g(|S_g|))$.
	
	Note that for $x\in S_f$, and $y\not\in S_f$, then $\E\big[ f(\xvec)|I=x \big] \leq \E\big[ f(\xvec)|I=y \big]$, and the same holds true for $g$ and $S_g$. We will re-order the rows of $A$ in the following way: $\sigma_f(1),\hdots,\sigma_f(|S_f|)$, followed by all indices not within either $S_f$ or $S_g$, followed by $\sigma_g(|S_g|),\hdots,\sigma_g(1)$. Then $h_1$ will be monotone increasing and $h_2$ will be monotone decreasing, thus showing $Cov(f(\xvec), g(\xvec))\leq 0$. An almost identical proof will show this true for $f,g$ both monotone decreasing.
\end{proof}
\begin{corollary} \label{cor:rankone-NA}
  If $\vc{u},\, \vc{v} \in (\reals_{\geq 0})^N$ are non-negative vectors,
  then the random variables $X_1,X_2,\ldots,X_N$ defined by
  sampling a uniformly random permutation $\pi : [N] \to [N]$
  and setting $X_i = u_i v_{\pi(i)}$ are negatively associated.
\end{corollary}
\begin{proof}
  The matrix $A = \vc{u} \vc{v}^{\trans}$ has non-negative entries
  and consistently ordered rows, so we may apply 
  \Cref{lem:constorderedrows} to deduce the corollary.
\end{proof}
\begin{corollary} \label{cor:permut-NA}
  Let $\mathscr{X} = \{x_1,x_2,\ldots,x_m\}$ be any multiset of non-negative numbers, 
  and for some $n \leq m$ let $X_1, X_2, \ldots, X_n$
  denote random variables obtained by drawing $n$
  uniformly random samples {\em without replacement} from $\mathscr{X}$.
  (In other words, the conditional distribution of $X_i$
  given $X_1,\ldots,X_{i-1}$ is uniform over the multiset
  $\mathscr{X} \setminus \{X_1,\ldots,X_{i-1}\}$.)
  Then $X_1,\ldots,X_n$ are negatively associated. 
\end{corollary}
\begin{proof}
  The special case $n = m$, in which the variables
  $X_1, \ldots, X_m$ constitute a random permutation
  of the elements of $\mathscr{X}$, can be obtained
  from \Cref{cor:rankone-NA} by setting
  $\vc{u} = (x_1,x_2,\ldots,x_m)^{\trans}$ and
  $\vc{v} = (1,1,\ldots,1)^{\trans}$. The general
  case in which $n \leq m$ can then be obtained
  by observing that the property of negative
  association is preserved under taking subsets
  of a set of random variables.
\end{proof}

We will be making use of the following Chernoff bound
for negatively associated random variables.
\begin{restatable}{lemma}{lemchernoffNA}\label{lem:chernoff-NA}
Suppose $X_1,...,X_N$ are negatively associated variables for which $X_i\in[0,1]$ always, and $\E[\sum_i X_i]=\mu$. Then Chernoff's multiplicative tail bound holds. That is, for any $\gamma > 0$, 
\begin{align} \label{eq:upper-tail-NA}
  \Pr\left[ \sum_i X_i \geq e^{\gamma} \mu \right] & \leq
  \left[ \exp \left( e^{\gamma} - 1 - \gamma e^{\gamma} \right) \right]^\mu < e^{-\frac{1}{2} \gamma^2 \mu} \\
  \label{eq:lower-tail-NA}
  \Pr\left[ \sum_i X_i \leq e^{-\gamma} \mu \right] & \leq
  \left[ \exp \left( e^{-\gamma} - 1 + \gamma e^{-\gamma} \right) \right]^{\mu} .
\end{align}
Furthermore, when $0 < \gamma < \frac12$ the second inequality implies
\begin{equation} \label{eq:lt-NA-restated}
  \Pr\left[ \sum_i X_i \leq e^{-\gamma} \mu \right] \leq
  e^{-\frac13 \gamma^2 \mu} .
\end{equation}
\end{restatable}
The Chernoff bound is often expressed in terms of the tail probabilities
$\Pr\left[ \sum_i X_i \geq (1+\delta) \mu \right]$
and $\Pr\left[ \sum_i X_i \leq (1-\delta) \mu \right]$,
with the bound on the right side of the inequality then being written
as a function of $\delta$.
For a proof, see~\cite{dubhashi1996balls,neg-assoc-notes}. 
The version of the Chernoff bound stated above is obtained from the
usual one by substituting $\gamma = \ln(1 + \delta)$ in the first inequality and
$\gamma = - \ln(1-\delta)$ in the second.
The inequality $- e^{\gamma} + 1 + \gamma e^{\gamma} \geq \frac12 \gamma^2$
is derived by writing it in the equivalent form $\int_0^\gamma t e^t \, dt
\geq \int_0^\gamma t \, dt$ and comparing integrands. The inequality
$-e^{-\gamma} + 1 - \gamma e^{-\gamma} \geq \frac13 \gamma^2$ is justified
by using Taylor's Theorem to deduce that the left side is bounded
below by $\frac12 \gamma^2 - \frac13 \gamma^3$ for $0 < \gamma < 1$
and then noting that $\frac12 \gamma^2 \geq \frac13 \gamma^2 + \frac13 \gamma^3$
when $0 < \gamma < \frac12$. 

As a first application of \Cref{lem:chernoff-NA} we can prove
the first tail bound asserted in \Cref{thm:tailbound}.
\begin{lemma} \label{lem:tailbound-single}
  Suppose $\vc{u}, \, \vc{v} \in (\reals_{\geq 0})^N$ are
  non-zero, non-negative vectors satisfying
  \begin{equation*} 
    \left( \frac{\|\vc{u}\|_1}{\|\vc{u}\|_{\infty}} \right)
    \left( \frac{\|\vc{v}\|_1}{\|\vc{v}\|_{\infty}} \right)
    \; \geq \; C N
    \tag{\ref{eq:tailbound-hyp}}
  \end{equation*}
  Suppose $D$ is a doubly stochastic matrix defining
  a bilinear form $B(\cdot,\cdot)$ via
    \begin{equation*} 
      B(\vc{x},\vc{y}) = \sum_{i \neq j} D_{ij} x_i y_j .
      \tag{\ref{eq:bilinear-form}}
    \end{equation*}
  Let $M = 1$ if $D$ is a permutation matrix, and $M = N^2$ otherwise.
  If $P$ is a uniformly random $N$-by-$N$ permutation matrix
  then for any $\gamma > 0$,
  \begin{equation*} 
      \Pr \left( B(\vc{u},P \vc{v})  \, \geq \, e^{\gamma} \frac{\|\vc{u} \|_1 \, \|\vc{v}\|_1}{N} \right)
      \leq M e^{-\frac12 \gamma^2 C} .
      \tag{\ref{eq:tailbound-single}}
  \end{equation*}
\end{lemma}
\begin{proof}
  The Birkhoff-von Neumann Theorem says that $D$ can be expressed as a convex combination
  of permutation matrices, and Carath\'{e}odory's Theorem says that there exists such an
  expression in which the number of constituent permutation matrices is at most
  $(N-1)^2 + 1$, which is bounded above by $N^2$. Hence, $D$ can be expressed as a
  convex combination of at most $M$ permutation matrices, where $M$ is
  defined as in the lemma statement. 
  The bilinear form $B$ is thus a convex combination of at most
  $M$  bilinear forms $B_{\sigma}$, where $B_{\sigma}$ is defined
  for a permutation $\sigma$ by
  \[
  B_{\sigma}(\vc{u},\vc{v}) = \sum_{i : \sigma(i) \neq i} u_i v_{\sigma(i)} .
  \]
  We will prove the special case of the lemma when $D$ is a
  permutation matrix and $B = B_{\sigma}$ for some $\sigma$;
  the general case will then follow by the union bound.

  If $\tau$ is the random permutation such that $P_{i,\tau(i)}=1$ for all $i$, then for any
  permutation $\sigma$ the composition $\pi = \tau \circ \sigma$ is uniformly distributed
  over all permutations of $[N]$. Consequently, by \Cref{cor:rankone-NA},
  the random variables $X_i = \frac{u_i v_{\pi(i)}}{\|\vc{u}\|_{\infty} \| \vc{v}\|_{\infty}}$
  are negatively associated. By construction, they take values between 0 and 1. 
  Furthermore, the expected value of $\sum_{i=1}^N X_i$ can be computed by
  linearity of expectation, using the fact that the event $\pi(i)=j$ has
  probability $\frac1N$ for all $j$. 
  \[
  \mu = \E \left[ \sum_{i=1}^N X_i \right] =
  \frac1N \sum_{i=1}^N \sum_{j=1}^N \frac{u_{i} v_j}{\|\vc{u}\|_{\infty} \| \vc{v}\|_{\infty}}
  = \frac1N \frac{\| \vc{u} \|_1 \| \vc{v} \|_1}{\|\vc{u}\|_{\infty} \| \vc{v}\|_{\infty}}
  \geq C .
  \]
  Applying \Cref{lem:chernoff-NA}, the probability that
  $\sum_{i=1}^N X_i$ exceeds $\frac{e^{\gamma}}{N}
  \frac{\|\vc{u} \|_1 \, \|\vc{v}\|_1}{\|\vc{u}\|_{\infty} \vc{u}\|_{\infty}}$
  is less than $e^{- (1/2) \gamma^2 C}$. Inequality~\eqref{eq:tailbound-single}
  follows because $B_{\sigma}(\vc{u},\vc{v}) \leq {\|\vc{u}\|_{\infty} \| \vc{v}\|_{\infty}}
  \sum_{i=1}^N X_i$. 
\end{proof}
\begin{rmk} \label{rmk:tailbound-double-isnt-NA}
  After seeing the proof of the tail bound~\eqref{eq:tailbound-single},
  it is tempting to try proving an analogous tail bound
  for $B(P \vc{u}, P \vc{v})$ 
  using the random variables $X_1,\ldots,X_N$ defined by
  $$ X_i = \frac{u_{\tau(i)} v_{\tau(\sigma(i))}}{\|\vc{u}\|_{\infty} \| \vc{v}\|_{\infty}} . $$
  The trouble is that these random variables may
  fail to be negatively associated. As a simple example, suppose $\vc{u} = \vc{v} =
  (1,1,0,0)^{\trans}$ and let $\sigma = (1 \; 2)(3 \; 4)$ be the permutation of $\{1,2,3,4\}$ that
  transposes the first and last pairs of elements. Then
  $X_1 = u_{\tau(1)} v_{\tau(2)}$ and $X_2 = u_{\tau(2)} v_{\tau(1)}$.
  When $\tau(\{1,2\}) = \{1,2\}$ we have $X_1 = X_2 = 1$,
  and otherwise $X_1 = X_2 = 0$. Hence, $\E[X_1 X_2] = \frac16
  > \E[X_1] \E[X_2]$, violating negative association.
\end{rmk}
Despite the counterexample in \Cref{rmk:tailbound-double-isnt-NA}, we will
still be able to prove a tail bound for $B(P \vc{u}, P \vc{v})$
using negative association and the Chernoff bound,
however we will need to pursue a more indirect strategy.
We begin with the following tail bound
for random submatrices of a non-negative rank-one matrix.
\begin{lemma} \label{lem:submatrix-tail}
  Suppose $\vc{u}, \, \vc{v} \in (\reals_{\geq 0})^N$ are
  non-zero, non-negative vectors satisfying~\eqref{eq:tailbound-hyp}.
  For any $K \leq N/2$ let $(Q, R)$ denote a uniformly random sample
  from the set of ordered pairs of $K$-element subsets of $[N]$ that
  are disjoint from one another.
  Then for $0 < \gamma < 1$,
  \begin{align} \label{eq:submatrix-upper-tail}
    \Pr \left( \sum_{i \in Q} \sum_{j \in R} u_i v_j \, \geq \, 
    e^{\gamma} \frac{K^2}{N^2} \| \vc{u} \|_1 \| \vc{v} \|_1 \right)
    \; \leq \; 2 e^{- \frac18 \gamma^2 C K / N} \\
    \label{eq:submatrix-lower-tail}
    \Pr \left( \sum_{i \in Q} \sum_{j \in R} u_i v_j \, \leq \, 
    e^{-\gamma} \frac{K^2}{N^2} \| \vc{u} \|_1 \| \vc{v} \|_1 \right)
    \; \leq \; 2 e^{- \frac{1}{12} \gamma^2 C K / N} .
  \end{align}
\end{lemma}
\begin{proof}
  If $\sum_{i \in Q} \sum_{j \in R} u_i v_j \, \geq \, 
  e^{\gamma} \frac{K^2}{N^2} \| \vc{u} \|_1 \| \vc{v} \|_1$
  then at least one of the inequalities
  \begin{align} \label{eq:smt.1}
  \sum_{i \in Q} \frac{u_i}{\|\vc{u}\|_{\infty}} & \geq
  e^{\gamma/2} \frac{K}{N} \frac{\|\vc{u}\|_1}{\|\vc{u}\|_{\infty}} \\
  \label{eq:smt.2}
  \sum_{j \in R} \frac{v_i}{\|\vc{v}\|_{\infty}} & \geq
  e^{\gamma/2} \frac{K}{N} \frac{\|\vc{v}\|_1}{\|\vc{v}\|_{\infty}} 
  \end{align}
  is satisfied.
  Similarly, if $\sum_{i \in Q} \sum_{j \in R} u_i v_j \, \leq \, 
  e^{-\gamma} \frac{K^2}{N^2} \| \vc{u} \|_1 \| \vc{v} \|_1$
  then at least one of the inequalities
  \begin{align} \label{eq:smt.3}
  \sum_{i \in Q} \frac{u_i}{\|\vc{u}\|_{\infty}} & \leq
  e^{-\gamma/2} \frac{K}{N} \frac{\|\vc{u}\|_1}{\|\vc{u}\|_{\infty}} \\
  \label{eq:smt.4}
  \sum_{j \in R} \frac{v_i}{\|\vc{v}\|_{\infty}} & \leq
  e^{-\gamma/2} \frac{K}{N} \frac{\|\vc{v}\|_1}{\|\vc{v}\|_{\infty}} 
  \end{align}
  is satisfied. 
  To bound the probabilities of these events,
  let $X_1, X_2, \ldots, X_K$ be random variables obtained
  by drawing $K$ uniformly random samples without replacement
  from the multiset
  $\{ \frac{u_i}{\| \vc{u} \|_{\infty}} \, \mid \, 1 \le i \le n \}$
  and observe that $X_1 + \cdots + X_K$ and
  $\sum_{i \in Q}  \frac{u_i}{\|\vc{u}\|_{\infty}}$ are
  identically distributed. By \Cref{cor:permut-NA} the
  random variables $X_1,\ldots,X_K$ are negatively
  associated, by construction they are $[0,1]$-valued,
  and by linearity of expectation their sum has
  expected value $\frac{K}{N} \frac{\|\vc{u}\|_1}{\|\vc{u}\|_{\infty}}$.
  The assumption that $\left( \frac{\|\vc{u}\|_1}{\|\vc{u}\|_{\infty}}
  \right) \left( \frac{\|\vc{v}\|_1}{\|\vc{v}\|_{\infty}} \right)
  \geq CN$, combined with the inequality
  $\frac{\|\vc{v}\|_1}{\|\vc{v}\|_{\infty}} \leq N$,
  implies $\frac{K}{N}
  \frac{\|\vc{u}\|_1}{\|\vc{u}\|_{\infty}} \geq C K / N$.
  The Chernoff bound now implies that the probability of 
  inequality~\eqref{eq:smt.1} being satisfied is
  at most $e^{-\frac18 \gamma^2 C K / N}$, and the probability
  of inequality~\eqref{eq:smt.3} being satisfied is
  at most $e^{-\frac{1}{12} \gamma^2 C K / N}$.
  A similar argument using $K$ random variables
  drawn without replacement from the multiset
  $\{ \frac{v_i}{\| \vc{v} \|_{\infty}} \, \mid \, 1 \le i \le n \}$
  establishes that the probabilities of inequalities~\eqref{eq:smt.2}
  and~\eqref{eq:smt.4} being satisfied are bounded above
  by $e^{-\frac{1}{8} \gamma^2 C K / N}$ and
  $e^{-\frac{1}{12} \gamma^2 C K / N}$, respectively.
  The lemma now follows by applying the union bound.
\end{proof}

We are now ready to restate and prove \Cref{thm:tailbound}.

\thmtailbound*

\begin{proof}
  The first tail bound, inequality~\eqref{eq:tailbound-single},
  was already proven in \Cref{lem:tailbound-single}, so turn to
  proving~\eqref{eq:tailbound-double}. As in the proof
  of~\eqref{eq:tailbound-single}, we will be using
  the Birkhoff-von Neumann Theorem, Carath\'{e}odory's
  Theorem, and the union bound to reduce to the case where
  the doubly stochastic matrix $D$ is a permutation matrix.
  Accordingly, for the remainder of the proof we will be focused
  on a fixed permutation $\sigma$ and its associated
  bilinear form
  \[
    B_{\sigma}(\vc{x},\vc{y}) = \sum_{i : \sigma(i) \neq i} x_i y_{\sigma(i)} ,
  \]  
  and our goal will be to prove the tail bound~\eqref{eq:tailbound-double}
  when $B = B_{\sigma}$ and $M=1$.

  To start, we note that it is without loss of generality
  to assume that $\sigma$ has at most one fixed point. This
  is because if $F$ is the fixed-point set of $\sigma$ and
  $|F| > 1$, then we can compose $\sigma$ with a
  permutation whose fixed-point-set is the complement
  of $F$, to obtain a fixed-point-free permutation
  $\tilde{\sigma}$ that agrees with $\sigma$ on the
  complement of $F$. For every pair of non-negative vectors $\vc{x}, \vc{y}$
  we have $B_{\tilde{\sigma}}(\vc{x},\vc{y})
  \geq B_{\sigma}(\vc{x},\vc{y})$, so an upper
  bound on the probability of
  $B_{\tilde{\sigma}}(P \vc{u}, P \vc{v}) \, \geq \, e^{\gamma} \frac{\|\vc{u} \|_1 \, \|\vc{v}\|_1}{N}$
  will suffice to
  prove an upper bound on the probability of the same
  event for $B_{\sigma}$. Henceforth we will ignore the
  distinction between $\sigma$ and $\tilde{\sigma}$ and
  we'll simply assume that $\sigma$ has at most one
  fixed point. Let $N^*$ denote the complement of $F$
  in $[N]$, i.e.~$N^* = \{ i \mid \sigma(i) \neq i\}$.

  Define the \emph{cycle diagram} of $\sigma$ to be the
  directed graph with vertex set $N^*$ and edge set
  $\{(i,\sigma(i)) \mid i \in N^*\}$, which is a disjoint
  union of directed cycles. 
  The next step of the proof is to define a
  \emph{balanced 3-coloring} 
  $\chi : N^* \to \{0,1,2\}$
  of the cycle diagram of $\sigma$, by which we mean
  a proper 3-coloring such that each color is used
  at least   $\lfloor |N^*| / 3 \rfloor$ and at
  most $\lceil |N^*| / 3 \rceil$ times.
  We will then break down the bilinear form
  $B_{\sigma}$ as a sum 
  $\bsq[0] + \bsq[1] + \bsq[1]$ 
  where for $q \in \{0,1,2\}$,
  \[ \bsq(\vc{u}, \, \vc{v}) =
  \sum_{i : \chi(i) = q} u_i v_{\sigma(i)},
  \]
  and we will 
  prove exponential tail bounds
  for each of the quantities
  $\bsq(P \vc{u}, P \vc{v})$.
  The purpose of the 3-coloring
  is to allow us to condition on
  an event that breaks up dependencies
  such as the one identified in
  \Cref{rmk:tailbound-double-isnt-NA},
  enabling the use of negative association
  and Chernoff bounds.
  
  One can find a balanced coloring of the cycle
  diagram of $\sigma$ by a greedy strategy, 
  combining the following two simple
  observations.
  \begin{enumerate}
  \item \emph{Every directed 2-cycle
    has a balanced 3-coloring.}
    If the cycle has length $k$, then 
    color the $i^{\mathrm{th}}$ vertex
    of the cycle with $\chi(i) = i \pmod{3}$
    unless $k \equiv 1 \pmod{3}$, in
    which case the first $k-1$ vertices
    of the cycle are colored using
    $\chi(i) = i \pmod{3}$ and the last
    vertex is colored with the unique
    color that differs from both of its
    neighbors' colors.
  \item \emph{The disjoint union of two
    graphs with balanced 3-colorings
    also has a balanced 3-coloring.}
    If a balanced 3-coloring of a graph
    with $n$ vertices, let us say that
    a color is \emph{overused} if it is
    used more than $\lfloor n/3 \rfloor$
    times. If graph $G$ is the disjoint
    union of $G_0$ and $G_1$, each of
    which has a balanced 3-coloring,
    let $k_0$ and $k_1$ denote the
    number of overused colors in $G_0$
    and $G_1$, respectively. If
    $k_0 + k_1 \leq 3$ then we can
    recolor $G_1$ if necessary so that its set of
    overused colors is disjoint from
    the set of overused colors in $G_0$.
    The union of the two colorings is
    then a balanced 3-coloring of $G$.
    If $k_0 + k_1 > 3$ then it must be
    the case that $k_0 = k_1 = 2$, in
    which case we can recolor $G_1$ if
    necessary so that each $q \in \{0,1,2\}$
    is overused in at least one of
    $G_0, \, G_1$, and exactly one
    color is overused in both. The union
    of the two colorings is then a
    balanced coloring of $G$.
  \end{enumerate}
  Having defined the coloring $\chi$
  we now focus on one specific color
  $q \in \{0,1,2\}$ and aim to prove
  a tail bound for $\bsq(P \vc{u}, \, P \vc{v})$
  when $\tau$ is a uniformly random permutation
  and $P$ is the permutation matrix with
  $P_{i, \tau(i)} = 1$ for all $i$.
  To do so we will define $I = \chi^{-1}(\{q\})$
  to be the set of indices $i \in N^*$ whose
  color is $q$, and we will condition on the
  random variable $Z = \tau|_{N^* \setminus I}$,
  the restriction of $\tau$ to indices whose
  color differs from $q$. Some useful observations
  are the following.
  \begin{enumerate}[label=\textbf{[O\arabic*]}]
  \item \label{obs:q} The set $Q = \tau(I)$ is uniquely
    determined by $Z$: it is equal to the
    complement of $\tau(N^* \setminus I)$
    in $N^*$.
  \item \label{obs:r} The set $R = \tau(\sigma(I))$ is also
    uniquely determined by $Z$. In fact,
    because $\chi(\sigma(i)) \neq \chi(i)$
    for all $i$, the set $\sigma(I)$ must be
    disjoint from $I$, so the value of
    $\tau(i)$ for each $i \in \sigma(I)$
    is determined by $Z$.
  \item \label{obs:qr} Let $K_q = |I|$. Since $I$ and $\sigma(I)$
    are disjoint $K_q$-element subsets of $N^*$
    and $\tau$ is a uniformly random permutation
    of $N^*$, the joint distribution of the pair
    of sets $(Q,R) = (\tau(I),\tau(\sigma(I))$
    is the uniform distribution on ordered pairs
    of disjoint $K_q$-element subsets of $N^*$.
  \item \label{obs:bij} Conditional on $Z$, the restriction
    of $\tau$ to $I$ is a uniformly random
    bijection between $I$ and  $Q$.
  \end{enumerate}
  Define a random variable $Y$ by
  \[
  Y = \sum_{j \in Q} \sum_{k \in R} u_j v_k
  \]
  and observe that the value of $Y$ is determined by $Z$,
  since $Z$ determines the sets $Q$ and $R$. 
  By Observation~\ref{obs:bij} and linearity of expectation we have
  \begin{equation} \label{eq:expect-bsq}
    \E[\bsq(P \vc{u},P \vc{v}) \, \mid \, Z] =
    \sum_{i \in I} \E[u_{\tau(i)} v_{\tau(\sigma(i))} \, \mid \, Z] =
    \frac{1}{K_q} \sum_{j \in Q} \sum_{k \in R} u_j v_k = \frac{Y}{K_q} .
  \end{equation}
  Our goal now turns to bounding the probabilities of the following ``bad events.''
  \begin{align*}
    \evnt_1 & = \left\{ Y \leq e^{-\frac57 \gamma} \frac{K_q^2}{N^2} \| \vc{u} \|_1 \| \vc{v} \|_1 \right\} \\
    \evnt_2 & = \left\{ Y \geq e^{\frac47 \gamma} \frac{K_q^2}{N^2} \| \vc{u} \|_1 \| \vc{v} \|_1 \right\} \\
    \evnt_3 &= \left\{ \bsq(P \vc{u}, P \vc{v}) \geq e^{\frac37 \gamma} \frac{Y}{K_q} \right\} .
  \end{align*}
  First, using \Cref{lem:submatrix-tail} and the inequality
  $K/N \geq \frac1N \lfloor (N-1)/3 \rfloor \geq 1/4$, we
  have
  \begin{equation} \label{eq:evnt12}
    \Pr(\evnt_1) \leq 2 e^{-\frac{1}{12} \left( \frac{5\gamma}{7} \right)^2 \frac{C}{4} }
    < 2 e^{-\frac{1}{100} \gamma^2 C}, \qquad
    \Pr(\evnt_2) \leq 2 e^{-\frac18 \left( \frac{4 \gamma}{7} \right)^2 \frac{C}{4} }
    < 2 e^{-\frac{1}{100} \gamma^2 C} .
  \end{equation}
  Next we turn to bounding the conditional probability
  $\Pr(\evnt_3 \setminus \evnt_1 \, \mid \, Z=z)$, for each
  value $z$ in the support of $Z$. Recall
  that the value of $Y$ is determined by $Z$, and the
  event $\evnt_1$ is determined by the value of $Y$.
  Hence, the values $z$ in the support of $Z$ may be
  partitioned into two sets: $\mathcal{Z}_0$ is the
  set of $z$ such that $\evnt_1$ does not occur when
  $Z=z$, and $\mathcal{Z}_1$ is the set of $z$ such
  that $\evnt_1$ occurs when $Z=z$. Obviously, for
  $z \in \mathcal{Z}_1$, $\Pr(\evnt_1 \, \mid \, Z=z) = 1$
  so $\Pr(\evnt_3 \setminus \evnt_1 \, \mid \, Z=z) = 0$. 

  Assume henceforth that $z \in \mathcal{Z}_0$. Then
  $Y > e^{-\frac57 \gamma} \frac{K_q^2}{N^2} \| \vc{u} \|_1 \| \vc{v} \|_1$.
  Now, let $\vc{u}_Q$ denote the subvector of $\vc{u}$ indexed by the
  elements of $Q$, and let $\vc{v}_R$ denote the subvector of $\vc{v}$
  indexed by the elements of $R$. We will apply \Cref{lem:tailbound-single}
  to this pair of vectors. Note that
  $\|\vc{u}_Q\|_1 \| \vc{v}_R \|_1 = Y$. Hence,
  \begin{align*} 
    \frac{ \|\vc{u}_Q\|_1 \| \vc{v}_R \|_1 }{ \|\vc{u}_Q\|_{\infty} \| \vc{v}_R \|_{\infty} }
    & \geq \frac{Y}{ \|\vc{u}\|_{\infty} \| \vc{v} \|_{\infty} } \\
    & > e^{-\frac57 \gamma} \frac{K_q^2}{N^2}
    \frac{ \| \vc{u} \|_1 \| \vc{v} \|_1 }{ \|\vc{u}\|_{\infty} \| \vc{v} \|_{\infty} }
    & \geq e^{-\frac57 \gamma} \frac{K_q^2}{N^2} \cdot C N
    & > \frac{e^{-\frac57} K_q}{N} \cdot C K_q > \frac{e^{-\frac57}}{4} C K_q > \frac{C}{9} K_q .
  \end{align*}
  By Observation~\ref{obs:bij}, the random variable $\bsq(P \vc{u}, P \vc{v})$
  can be calculated by sampling a uniformly random bijection $\pi$
  between $Q$ and $R$ and computing the sum
  $\sum_{i \in Q} u_{i} v_{\pi(i)}$. Hence, by 
  \Cref{lem:tailbound-single},
  \begin{equation} \label{eq:evnt3}
    \Pr(\evnt_3 \, | \, Z = z \in \mathcal{Z}_0) \leq
    e^{-\frac12 (\frac37 \gamma)^2 \frac{C}{9} } <
    e^{-\frac{1}{100} \gamma^2 C} .
  \end{equation}
  Combining the cases $z \in \mathcal{Z}_0$ and $z \in \mathcal{Z}_1$,
  we have proven that $\Pr(\evnt_3 \setminus \evnt_1 \, | \, Z) < e^{-\frac{1}{100} \gamma^2 C}$
  pointwise. Hence,
  \[
  \Pr(\evnt_3 \setminus \evnt_1) =
  \E_{Z} \left[ \Pr(\evnt_3 \setminus \evnt_1 \, | \, Z) \right] < e^{-\frac{1}{100} \gamma^2 C} .
  \]
  Now, by the union bound, we find that
  \[
  \Pr(\evnt_2 \cup \evnt_3) \leq
  \Pr(\evnt_1 \cup \evnt_2 \cup \evnt_3)
  \leq \Pr(\evnt_1) + \Pr(\evnt_2) + \Pr(\evnt_3 \setminus \evnt_1)
  \leq 5 e^{-\frac{1}{100} \gamma^2 C} .
  \]
  On the complement of $\evnt_2 \cup \evnt_3$, we have the inequalities
  \begin{equation} \label{eq:good-q}
  \bsq(P \vc{u}, P \vc{v}) < e^{\frac35 \gamma} \frac{Y}{K_q} < e^{\frac35 \gamma} \cdot e^{\frac25 \gamma}
  \cdot \frac{1}{K_q} \cdot \frac{K_q^2}{N^2} \cdot \| \vc{u} \|_1 \| \vc{v} \|_1 =
  e^{\gamma} \frac{K_q}{N} \cdot \frac{  \| \vc{u} \|_1 \| \vc{v} \|_1 }{N} .
  \end{equation}
  With probability at least $1 - 15 e^{-\frac{1}{100} \gamma^2 C }$, the
  event $\evnt_2 \cup \evnt_3$ does not occur for any $q \in \{0,1,2\}$.
  In that case,
  \begin{equation} \label{eq:all-good-q}
    B_{\sigma}(P \vc{u}, P \vc{v}) =
    \bsq[0](P \vc{u}, P \vc{v}) + \bsq[1](P \vc{u}, P \vc{v}) + \bsq[2](P \vc{u}, P \vc{v})
    < e^{\gamma} \frac{K_0 + K_1 + K_2}{N} \cdot \frac{  \| \vc{u} \|_1 \| \vc{v} \|_1 }{N}
    \leq e^{\gamma} \frac{  \| \vc{u} \|_1 \| \vc{v} \|_1 }{N} .
  \end{equation}
  Hence, the negation of this inequality occurs with probability
  at most $15 e^{-\frac{1}{100} \gamma^2 C}$, as claimed. 
\end{proof}

%% file: semi-orn-design.tex
\section{Upper Bound: Semi-Oblivious Design} \label{sec:semi-obliv-design}

In this section we prove \Cref{thm:tradeoff}.\ref{to.sorn.whp}, restated below.

\noindent \textbf{\Cref{thm:tradeoff}.\ref{to.sorn.whp}.} 
\textit{ Given any fixed throughput value $r\in(0,\frac{1}{2}]$, let $\semih = \semih(r) = \lfloor\frac{1}{r}-1\rfloor$, 	and let }
\begin{align*}
    L_{upp}(r,N) & = \semih N^{1/\semih} .
\end{align*}
\textit{Then assuming $\frac{1}{r}\not\in\mathbb{Z}$, there exists a family of distributions over semi-oblivious reconfigurable network designs for infinitely many network sizes $N$ which attains maximum latency $\bigotilde(L_{upp}(r,N))$ with high probability (and in expectation) over time-stationary demands, and achieves throughput $r$ with probability 1. }
\vspace{3mm}

Similar to \Cref{sec:upper-bound}, we will begin by constructing an SORN design $\mathcal{S}^0$ which is parameterized by $N$, $\semih$, and $C$, where $C$ is a parameter which we set during our analysis to achieve the appropriate tradeoffs between throughput and latency.
We will then analyze $\mathscr{S}_N(\semih,C)$, a distribution over all SORN designs $\mathcal{S}^\tau$ which are equivalent to $\mathcal{S}^0$ up to re-labeling of nodes, and show that it satisfies the conclusion of \Cref{thm:tradeoff}.3. 
Before we define $\mathcal{S}^0$, we first provide some intuition behind the design.

\begin{dfn} \label{def:c-h-frame}
	A {\em $\chframe$} in $\mathbb{F}_p^\semih$ is a sequence of $C(\semih+1)$ vectors for which the following property holds. Any set of $\semih$ distinct vectors 
	forms a basis over the vector space $\mathbb{F}_p^\semih$.
\end{dfn}

The ORN design described in \Cref{sec:upper-bound} was defined using phases of Vandermonde vectors. 
This was only done to achieve the property that any set of $\semih$ vectors, each chosen from a different phase block, formed a basis over $\mathbb{F}_p^\semih$. 
No other special property of Vandermonde vectors was required.
Thus, using any $\chframe$ gives the same throughput-latency tradeoffs found in \Cref{thm:tradeoff}.

In order to guarantee throughput $r$ rather than achieve it with high probability, we need to provide alternate routing paths in the low probability case that the network becomes congested.
We will do this by rotating through a series of different $\chframe$s, so that in an entire period of the schedule, each node is directly connected to most other nodes an equal number of times. 
Our alternate paths will then use a simple 2-hop Valiant load balancing (VLB) routing strategy.

\begin{restatable}{lemma}{lemframetwisting}\label{lem:frame-twisting}
Suppose $A\in\mathbb{F}_p^{\semih\times \semih}$ is an invertible matrix, and $\mathcal{V} = (v_1,v_2,\hdots,v_{C(\semih+1)})$ is a  $(C,g)$-constellation in $\mathbb{F}_p^\semih$. 
Then the sequence $A\mathcal{V} = (Av_1,Av_2,\hdots,Av_{C(\semih+1)})$ is also a $(C,g)$-constellation in $\mathbb{F}_p^\semih$.
\end{restatable}

\begin{proof}
	Suppose not, that there exists some set of vectors $w_{i_1},\hdots w_{i_\semih}$ each from different blocks of $A\mathcal{V}$ which are linearly dependent. 
	Then WLOG there exists constants $\alpha_1,\hdots,\alpha_{\semih-1}$ such that $\alpha_1w_{i_1}+\hdots+\alpha_{\semih-1}w_{i_{\semih-1}} = w_{i_\semih}$.
	Then $\alpha_1Av_{i_1}+\hdots+\alpha_{\semih-1}Av_{i_{\semih-1}} = Av_{i_\semih}$ for vectors $v_{i_1},\hdots,v_{i_\semih}$ each from different blocks of $\mathcal{V}$. This is a contradiction due to distributivity of matrix and vector multiplication, and because $A$ is invertible.
\end{proof}

\subsection{Connection Schedule}\label{sec:semi-conn-sched}

We now move to defining the connection schedule of $\mathcal{S}^0$.
Consider the set of all diagonal invertible matrices $\mathcal{M}$, and let two matrices $M_1,M_2$ be related by $\sim$ if they are scalar multiples of one another. That is, if there is some scalar $a\in\mathbb{F}_p$ such that $M_1=aM_2$.
Let $\mathcal{A}\subset\mathcal{M}$ contain one representative from each of the equivalence classes of $\sim$. (Note that therefore, $|\mathcal{A}|=(p-1)^{g-1}$.)
Also let $\mathcal{V}$ be any sequence of $C(\semih+1)$ distinct Vandermonde vectors not including the vector $(1,0,\hdots,0)$. 
Order $\mathcal{V}$ arbitrarily, so that $\mathcal{V} = \{ \bm{v}_0,\bm{v}_1,\hdots,\bm{v}_{C(\semih+1)-1} \}$.

Then by \Cref{lem:frame-twisting}, $A\mathcal{V}$ is a $\chframe$ for any matrix $A\in\mathcal{A}$.
Order the set of matrices $\mathcal{A}$ arbitrarily, so that $\mathcal{A} = \{A_0,A_1,\hdots,A_{(p-1)^{\semih-1}-1}\}$. 
We rotate through the $\chframe$s formed by matrices in $\mathcal{A}$ to achieve our connection schedule.

More formally, we set the period length of the
  schedule to be $T = (p-1)^{\semih-1} C (\semih+1) (p-1) = (p-1)^{\semih} C (g+1) < C(g+1) N$, and we identify each congruence class
    $k \pmod{T}$ with 
 a constellation number $f$, a phase number $x$ and a scale factor $s$,
for which $0\leq f\leq (p-1)^{\semih-1}-1$, $0\leq x < C (\semih+1)$, and $1\leq s<p$, such that $k = C(\semih+1)(p-1) f + (p-1) x + s - 1$. 
It is useful to think of timesteps as 3-tuples, $k=(f,x,s)$, so we will sometimes abuse notation and refer to timestep $(f,x,s)$ in the sequel, when we mean $k = C(\semih+1)(p-1) f + (p-1) x + s - 1$.
The connection schedule of $\mathcal{R}^0$, during timesteps $t \equiv k \pmod{T}$, uses permutation
$\pi_{k}^0(a) = a + s A_{f}\bm{v}_x$,
where $f,x$ and $s$ are the constellation number, phase number, and scale associated to $k$.

As described above, $\mathscr{S}_N(\semih,C)$ is a distribution over all SORN designs $\mathcal{S}^\tau$ which are equivalent to $\mathcal{S}^0$ up to re-labeling.
When we sample a random design $\mathcal{S}^\tau$, we sample a uniformly random permutation of the node set $\tau:\mathbb{F}_p^h\rightarrow \mathbb{F}_p^\semih$, producing the schedule $\pi_{k}^\tau(a) = \tau^{-1} \Big(\pi_{k}^0\big(\tau(a) \big) \Big)$.
Note that, for every edge from node $a$ to node $\pi_t^0(a)$ in $\mathcal{S}^0$, there is a unique equivalent edge from $\tau(a)$ to $\tau(\pi_t^0(a))$ in $\mathcal{S}^\tau$.

\subsection{Routing Protocol} \label{sec:semi-routing-scheme}

The routing protocol $\{S_\sigma^0 : \sigma \mbox{ permut on } [N]\}$ will, for each $\sigma$, use one of two types of routing paths.
The first type is the $(\semih+1)$-hop paths that we wish to route on. 
For most $\sigma$, routing on these paths will not overload any edges in the network. Thus, for those $\sigma$, $S_\sigma^0$ will include only those such paths.

However, with low probability over $\sigma$, routing on these paths will cause too much congestion on some edge in the network to be used. 
In this case, we will designate an alternate set of paths for $S_\sigma^0$ to use.
The alternate set of paths will take only 2 hops in the network, and will suffer significantly higher maximum latency.
However, we will show that since this is a low probability event over choice of $\sigma$, this will not meaningfully increase our average latency.

To route from node $a$ to node $b$ starting at timestep $t$, first delay until a new $\chframe$ $A\mathcal{V}$ begins.

\textbf{$(\semih+1)$-hop paths.} \hspace{1mm} 
In this case, we use the same distribution over routing paths as in \Cref{sec:orn-rout-scheme}, when considering the set of $C(g+1)$ phases all belonging to the $\chframe$ beginning after time $t$.
Due to the added delay, paths of this type have maximum latency $2C(\semih+1)N^{1/h}$, instead of the maximum latency cited in \Cref{sec:orn-rout-scheme}.

\textbf{2-hop paths.} \hspace{1mm}
To describe the distribution over 2-hop paths, first consider the following.
Given a fixed Vandermonde vector $\bm{v}\in\mathcal{V}$, consider the set of edges formed by $a \rightarrow a + s A \bm{v} = b$ for all scalar factors $s$ and matrices $A\in\mathcal{A}$.
Note that an edge between any node pair $a,b$ for which the vector $b-a$ has only non-zero coordinates appears exactly once in this set.
This is because $A\in\mathcal{A}$ contains all invertible diagonal matrices which are not scalar multiples of each other. 
Additionally, an edge between $a,b$ never appears if the vector $b-a$ has any coordinates equal to zero. (Recall the vector $ (1,0,\hdots,0) \not\in \mathcal{V}$.)
Then across the entire period, an edge between any node pair $a,b$ for which the vector $b-a$ has only non-zero coordinates appears exactly $C(\semih+1)$ times, once for each $\bm{v} \in \mathcal{V}$.

Consider the following random process for choosing a 2-hop path from $a$ to $b$.
Uniformly at random, choose a node $b'$ for which both $b'-a$ and $b-b'$ have only non-zero coordinates.
Also uniformly at random, choose Vandermonde vectors $v_a, v_b \in \mathcal{V}$. 
Compute the unique invertible diagonal matrices $A_a, A_b\in\mathcal{A}$ and scalar factors $s_a, s_b\in\{1, \hdots, p-1\}$ for which $b'-a = s_a A_a v_a$ and $b-b'=s_b A_b v_b$.
Over the next full period of the schedule, or $(p-1)^{\semih-1} C(\semih+1) (p-1)$ timesteps, take the direct hop from $a$ to $b'$ which appears during the $\chframe$ $A_a\mathcal{V}$. 
Wait for the period to finish. 
Then during the next period, take the hop from $b'$ to $b$ which appears during the $\chframe$ $A_b\mathcal{V}$.

Note that paths of this type always take both hops during consecutive distinct periods, or iterations, of the schedule.
Thus, paths of this type will have maximum latency $2(p-1)^{\semih} C(\semih+1) + C(\semih+1)(p-1) \leq C(\semih+1)N^{1/h} + 2C(\semih+1)N \leq \bigotilde(N)$.

If routing $rD_\sigma$ on $(\semih+1)$-hop paths does not overload edges in the network, then $S_\sigma$ routes all demand between $a,\sigma(a)$ pairs on $(\semih+1)$-hop paths. 
Otherwise, if routing $rD_\sigma$ on $(\semih+1)$-hop paths would overload some edge in the network, then $S_\sigma$ routes all demand between $a,\sigma(a)$ pairs on 2-hop paths.

As written here, $S_\sigma^0$ must make one choice for all timesteps $t$: to either route on $(\semih+1)$-hop paths or 2-hop paths.
In \Cref{app:k-frame-cor}, we discuss how to analyze a design which allows $S_\sigma^0$ to route flow on a combination of $(\semih+1)$-hop and 2-hop paths, depending on starting timestep $t$.

To route over $\mathcal{S}^\tau$ for general $\tau$, note that the edges of $\mathcal{S}^\tau$ are in a bijection with $\mathcal{S}^0$. 
Thus, any path from node $a$ to node $b$ in $\mathcal{S}^\tau$ has a unique equivalent path from $\tau(a)$ to $\tau(b)$ in $\mathcal{S}^0$.
To define the routing protocol $\{S_\sigma^\tau:\sigma \mbox{ permut on } [N]\}$ in $\mathcal{S}^\tau$, simply apply this bijection to the routing paths from $\tau(a)$ to $\tau(\sigma(a))$ in $\{S_\sigma^0:\sigma \mbox{ permut on } [N]\}$.

\subsection{Throughput-Latency Tradeoff}\label{sec:semi-design-tradeoff}

\begin{restatable}{theorem}{thmsemilowprob}\label{thm:semi-low-prob}
Given a fixed throughput value $r$, let $\semih = \semih(r) = \lfloor\frac{1}{r}-1\rfloor$ and $\eps = \eps(r) = \semih + 1 - (\frac{1}{r}-1)$, and assume $\eps\neq 1$.
As $N$ ranges over the set of prime powers $p^{\semih}$ for primes $p$ exceeding $\max \left\{ C(\semih+1), 2 + \frac{2}{1-\eps}, \frac{g+3}{\eps}-2, \frac{2-\delta}{1-\delta} \right\}$ for $\delta=\frac{(g+1)^{1/g}} {(g+2-\eps)^{1/g}}$,
 let $\gamma=\ln \left( \frac{\semih + 2 - \eps}{\semih+1} \right)$ and $C=\frac{\log\log N}{\gamma^2}\ln(N)$, and let
\[ L_{upp} = gN^{1/g} \]
Then:
\begin{enumerate}
	\item the fixed SORN design $\mathcal{S}^0$ 
	guarantees throughput $r$ (with respect to stationary demands), and achieves maximum latency $\bigotilde(L_{upp})$ with high probability under the uniform distribution.
	\item the family of distributions $\mathscr{S}_N(\semih,C)$ guarantees throughput $r$, and achieves maximum latency $\bigotilde(L_{upp})$ with high probability.
\end{enumerate}
\end{restatable}
\vspace{3mm}

Note that if $\eps=1$, then $\frac{1}{r}\in\mathbb{Z}$, and there do not exist primes $p$ for which $p\geq 2+\frac{2}{1-\eps}$. Thus, we condition against $\eps=1$.

Both parts of this theorem will be proven by focusing on the probability that $S_\sigma^0$ must deviate from sending all flow on $(g+1)$-hop paths to sending all flow on 2-hop paths.
This is directly correlated with when congestion occurs on physical edges in the design $\mathcal{S}^0$, if we were to always send flow on $(g+1)$-hop paths.
We note the similarities between $\mathcal{S}^0$ and $\mathcal{R}^0$ from \Cref{sec:upper-bound}, and apply the same exponential tail bounds of bilinear sums to get our result.

\begin{proof}

First, let us confirm that the 2-hop ``failover scheme'' of $\mathcal{S}^{\tau}$ guarantees throughput $r$.
Fix some permutation demand $D_\sigma$ and an edge $e$, and consider for each demand pair $i,\sigma(i)$ how much much flow is crossing edge $e$ due to $i,\sigma(i)$ traveling on 2-hop paths. 
If 1st hop flow crosses edge $e$ from $i$ to $\sigma(i)$, then it must be the case that $tail(e)=i$ and both $head(e)-i$ and $\sigma(i)-head(e)$ have only non-zero coordinates.
Similarly, if 2nd hop flow crosses edge $e$ from $i$ to $\sigma(i)$, then $head(e)=\sigma(i)$ and both $tail(e)-i$ and $\sigma(i)-tail(e)$ have only non-zero coordinates.

Each demand pair $i,\sigma(i)$ contributes $rC(g+1)(p-1)^g$ total flow per period. 
For any node pair $i,\sigma(i)$, there are at least $(p-2)^g$ different nodes $b$ for which $b-i$ and $\sigma(1)-b$ both have only non-zero coordinates. 
And for each of these nodes $b$, there are exactly $C$ different phases which connect $i$ to $b$, and exactly $C$ different phases which connect $b$ to $\sigma(i)$.
Thus, the amount of first-hop flow traversing edge $e$ is no more than
$\frac{rC(g+1)(p-1)^g}{C(p-2)^g}$.
This is no more than 1 when $p\geq \frac{2-\delta}{1-\delta}$ for $\delta=\frac{(g+1)^{1/g}} {(g+2-\eps)^{1/g}}$, which we condition on in the statement of the theorem.

Thus, we focus on showing that $\mathcal{S}^0$ sends flow on $(g+1)$-hop paths with high probability over the uniform distribution.

Like before, we may assume without loss of generality that
  the demand matrix $D(t)$ is doubly stochastic
  for all $t$.

We first consider the failure probability of edges within each $\chframe$ individually.
Fix an edge $e$ and $0\leq q\leq \semih$, and consider the amount of flow traversing edge $e$ traveling on paths where edge $e$ occurs in the $(q+1)$-th phase block of the flow path. 

Note that, unlike in the proof of \Cref{thm:fixed-design-tradeoff}, edges $e$ that appear in the $(q+1)$th phase block of a $\chframe$, for $0\leq q\leq \semih$, will {\em only} have $(q+1)$-th hop flow traversing $e$, due to delaying flow before routing by whole $\chframe$s instead of single phase blocks.
Then the total amount of $(q+1)$-th hop flow traversing edge $e$ equals the total amount of any-hop flow traversing edge $e$. 
First we examine $q=0$.
First-hop flow traversing edge $e$ originates at source node $\tl(e)$ during the constellation preceding the one to which $e$ belongs.
There are $C(g+1)(p-1)$ time steps during that phase block, and $r$ units of flow per time step originate at $\tl(e)$. 
Each unit of flow is divided evenly among
  a set of at least $(p-2) C^{\semih+1}$ pseudo-paths, at
  most $C^{\semih}$ of which begin with edge $e$ as their
  first hop. 
(After fixing the first hop and the destination of a $(\semih+1)$-hop pseudo-path, the rest of the path is uniquely determined by the $g$-tuple of phases $x_2,\ldots,x_{g+1}$.) 
Hence, of the $r C (g+1) (p-1)$ units of flow that could traverse $e$ as their first hop, the fraction that actually do traverse $e$ as their first hop is at most $\frac{C^{\semih}}{(p-2) C^{\semih+1}}$.
Consequently, for an edge $e$ occurring in the first phase block of a $(C,g)$-constellation, the amount of first-hop flow on $e$ is bounded above by $\frac{r C (g+1) (p-1) \cdot C^{\semih}}{(p-2) C^{\semih+1}} = \left( \frac{p-1}{p-2} \right) (g+1) r .$
(Note that this is not a probabilistic statement; the upper bound on first-hop flow holds with probability 1.)
A symmetric argument shows that for an edge $e$ occurring in the last phase block of a $(C,g)$-constellation, the amount of last-hop flow on $e$ is bounded above by $\left( \frac{p-1}{p-2} \right) (g+1) r $ as well.

Now suppose $1\leq q\leq \semih-1$, and let $Y_i$ be the random variable realizing the amount of $(q+1)$-th hop flow traversing edge $e$ due to source node $i$, normalized by $\frac{1}{g+1}$. Clearly, the total amount of $(q+1)$-th hop flow traversing $e$ will be $(g+1) \sum_i Y_i$.
 The variables $Y_i$ act exactly as the random variables $X_i$ in \Cref{sec:ub-tradeoff}, in the proof of \Cref{thm:orn-tradeoff}.
 Therefore, the same tail bound conclusions about their sum
 are applicable.

Therefore, over the uniform distribution for the fixed design $\mathcal{S}^0$, and for the family of distributions $\mathscr{S}_N(\semih,C)$, we have
\begin{align*}
	& \Pr[e \text{ has }\geq (\semih+1) e^{\gamma} r \text{ flow when routing } (g+1) \text{-hop paths}] \leq 15 N^2 e^{- \frac1{200} \gamma^2 C } \\
	\implies & \Pr[\text{any edge } e \text{ has }\geq (\semih+1)e^\gamma r \text{ flow when routing } (g+1) \text{-hop paths}]\\
	&\leq (p-1)^{g-1}C(g+1)(p-1)N \cdot 15 N^2 \left( e^{-\frac{1}{200}\gamma^2} \right)^{C} \\
	&\leq 15 N^4 (g+1) \frac{\log\log N}{\gamma^2}\ln(N) e ^{-\frac{1}{200} \log\log N\ln(N)} \\
	&\leq \left( 15N^4 (g+1) \frac{\log\log N\ln(N)}{\gamma^2} \right)N^{- \frac{1}{200} \log\log N} \\
	&\leq \bigo\left(\frac{1}{\gamma^2 N^d}\right) \text{ for any constant }d.
\end{align*}

Finally, we need to show that if none of the bad events as described above occur, if every edge has at most $e^{\gamma} r$ $(q+1)$-th hop flow for $1 \leq q \leq g-1$, then no edge will be overloaded.

First, note that the amount of flow traversing edges $e$ during the first and last phase blocks of any constellation will be at most $\frac{p-1}{p-2}(g+1)r$. This is no more than 1 when $p\geq \frac{g+3}{\eps}-2$, which we conditioned on in the statement of the theorem.

Next, note that assuming no bad events occur, the amount of flow traversing edge $e$ occurring during any other phase block of any constellation must be at most
\begin{align*}
	(g+1)e^\gamma r = (g+1) \frac{g+2-\eps}{g+1} \frac{1}{g+2-\eps} = 1 .
\end{align*}

\end{proof}

Note that \Cref{thm:tradeoff}.3 is a direct corollary of \Cref{thm:semi-low-prob}.2.

\subsection{Provably Separating the Capabilities Between ORNs and SORNs} \label{sec:semi-obliv-provable-sep}

In this section we show that semi-oblivious routing has a provable asymptotic advantage over oblivious routing in reconfigurable networks. 
In order to do so, we must compare the guaranteed throughput versus latency tradeoffs achieved by the family of SORN designs $\mathscr{S}_N(\semih,C)$ described above and distributions over ORN designs.
We will show below that our family of SORN designs $\mathscr{S}_N(\semih,C)$ has a provable asymptotic advantage over ORNs in {\em average latency}. To do so, we provide the following lower bound on average (expected) latency of distributions over ORN designs.

\begin{restatable}{theorem}{thmlbinformal}\label{thm:obliv-avg-lat-lb}
	Consider any constant $r \in (0,\frac{1}{2}].$ Let $h = h(r) = \floor{\frac{1}{2r}}$ and $\eps_o = \eps_o(r) = h + 1 - \frac{1}{2r}$,
	and let $L_{obl}(r,N)$ be the function 
	
	\[ L_{obl}(r,N) = \eps_o(\eps_o N)^{1/h} + N^{1/(h+1)}  . \]
	
	Then for every $N > 1$ and every distribution of ORN designs $\mathscr{R}$ on $N$ nodes that guarantees throughput $r$, the expected \textbf{average} latency of $\mathcal{R}\sim\mathscr{R}$ is at least $\Omega(L_{obl}(r,N))$. 
\end{restatable}

The proof of \Cref{thm:obliv-avg-lat-lb} follows a similar structure as the lower bound proof of \cite{stoc-paper}, only with an added average latency constraint in the starting linear program, which results in an additional variable in the corresponding dual program, which must be reasoned about and assigned a value.
We leave the proof to \Cref{app:avg-lat-obliv}.

\begin{restatable}{theorem}{thmprovablesep}\label{thm:provable-sep}
	Consider any constant $r \in (0,\frac{1}{2}]$, and let $\semih = \semih(r) = \lfloor\frac{1}{r}-1\rfloor$ and $\eps = \eps(r) = \semih + 1 - (\frac{1}{r}-1)$
	Then if $r\in\left(0,\frac{1}{4}\right] \cup \left[\frac{1}{4-(2/N^{1/6})},\frac{1}{3}\right]$ and $\frac{1}{r}$ is not an integer, the family of SORN designs $\mathscr{S}_N(g,C)$ achieves asymptotically better average latency than any family of ORN designs which guarantees throughput $r$.
\end{restatable}
\begin{proof}
By \Cref{thm:obliv-avg-lat-lb}, any family of ORN designs which guarantees throughput $r$ must suffer average latency $\Omega\left( L_{obl}(r,N) \right)$.
Also recall that the family of SORN designs $\mathscr{S}_N(g,C)$ achieves maximum latency $\bigotilde(g N^{1/g})$ with high probability as long as $\frac{1}{r}$ is not an integer. This implies it also achieves average latency $\bigotilde(g N^{1/g})$, since with probability 1 it achieves maximum latency $\bigotilde(N)$.
We divide the set of throughput values $r\in\left(0,\frac{1}{4}\right] \cup \left[\frac{1}{4-(2/N^{1/6})},\frac{1}{3}\right]$ into the following cases.

\begin{enumerate} 
	\item $r\leq\frac{1}{5}$. 
	Then $\semih(r) > h(r)+1$. Since $L_{obl}(r,N)\geq N^{1/(h+1)}$, then $L_{obl}(r,N)$ is asymptotically greater than $\bigotilde(L_{upp})$.
	\item $r\in\left(\frac{1}{5},\frac{1}{4}\right]$. Then $\eps_o(r)= h + 1 - \frac{1}{2r} \geq \frac{1}{2}$, and $\semih(r) > h(r)$. Therefore, $L_{obl}(r,N)$ is asymptotically greater than $\bigotilde(L_{upp})$.
	\item $r\in\left[\frac{1}{4-(2/N^{1/6})},\frac{1}{3}\right]$.  Then $\eps_o(r) \geq \frac{1}{N^{1/6}}$, $g(r)=2$, and $h(r)=1$. So $\eps_o(\eps_o N)^{1/h} \geq \left(\frac{1}{N^{1/6}}\right)^2 N = N^{2/3}$. Additionally, $gN^{1/g} = 2 \sqrt{N}$. Therefore, $L_{obl}(r,N)$ is asymptotically greater than $\bigotilde(L_{upp})$. 
\end{enumerate}

\end{proof}

%% file: lower-bound.tex
\section{Lower Bound} \label{sec:lower-bound}

In this section we prove \Cref{thm:tradeoff}.4, restated below.\\

\noindent \textbf{\Cref{thm:tradeoff}.4.} \textit{
Given any fixed throughput value $r\in(0,\frac{1}{2}]$, let $\semih = \semih(r) = \lfloor\frac{1}{r}-1\rfloor$ and $\eps = \eps(r) = \semih + 1 - (\frac{1}{r}-1)$, and let 
\begin{equation*}
	L_{low}(r,N) = \semih \left((\eps N)^{1/\semih} + N^{1/(\semih+1)}\right) 
\end{equation*}
Then any fixed ORN design $\mathcal{R}$ of size $N$ which achieves throughput $r$ with high probability must suffer at least $\Omega(L_{low}(r,N))$ maximum latency. 
}
\vspace{3mm}

\begin{proof}

We will start by upper bounding throughput for a given maximum latency.
We begin with a set of $N!$ linear programs, one for each possible permutation $\sigma$ on the node set, to solve the following problem:
given maximum latency $L$ and some reconfigurable network schedule, LP$_\sigma$ finds a set of routing paths to route $r$ flow between each $a,\sigma(a)$ pair, and maximizes the value $r$ for which this is possible. 
\begin{center}\fbox{\begin{minipage}{0.75\textwidth}
	\textbf{Primal LP$_\sigma$}
	\vspace{-3mm}
	\[\begin{lparray}
    \mbox{maximize} & r \\
    \mbox{subject to} & \sum_{P\in\pths_L(a,\sigma(a),t)} \mathcal{R}_{\sigma}(a,t,P)  = r \hspace{9mm} \forall a\in[N], t\in[T] \\
    & \sum_{a,t}\sum_{P\in\pths_L(a,\sigma(a),t) : e\in P} \mathcal{R}_{\sigma}(a,t,P)  \leq 1  \hfill  \qquad \forall e \in \ephys \\
    &\mathcal{R}_{\sigma}(a,t,P)  \geq 0 \hfill \forall a\in[N], t\in[T], P\in \pths_L(a,\sigma(a),t)
  \end{lparray}
\] \end{minipage} }
\end{center}

We then take the dual program of each LP$_\sigma$ to find Dual$_\sigma$.

\begin{center}\fbox{\begin{minipage}{0.75\textwidth}
	\textbf{Dual$_\sigma$}
	\vspace{-3mm}
	\[\begin{lparray}
    \mbox{minimize} & \sum_e \beta_{\sigma e} \\
    \mbox{subject to} & \alpha_{at\sigma} \leq \sum_{e\in P} \beta_{\sigma e} \hspace{5mm}\hfill \forall a\in[N], t\in[T], P\in \pths_L(a,\sigma(a),t) \\
    & \sum_{at}\alpha_{at\sigma} \geq 1 \\
    & \beta_{\sigma e} \geq 0 \hfill \forall e\in\ephys
  \end{lparray}
\]\end{minipage} }
\end{center}

For each Dual$_\sigma$, we will define a dual solution. Then, we will analyze an upper bound on the objective value of Dual$_\sigma$, with high probability over the random sampling of $\sigma$.

We will also reframe Dual$_\sigma$ in the following way, which will be easier to work with. Note that $\left( \sum \beta_{\sigma e} \right) \big/\left( \sum \alpha_{at\sigma} \right)$ is still an upper bound on throughput.

\begin{center}\fbox{\begin{minipage}{0.45\textwidth}
	\vspace{-4mm}
	\[\begin{lparray}
    \mbox{minimize} & \left( \sum_e \beta_{\sigma e} \right) \big/
        \left( \sum_{at} \alpha_{at\sigma} \right) \\
    \mbox{subject to} & \alpha_{at\sigma} \leq \sum_{e\in P} \beta_{\sigma e} \hspace{5mm}\hfill \forall a,t,P \\
    & \beta_{\sigma e} \geq 0
  \end{lparray}
\]\end{minipage} }
\end{center}

To understand how we construct and analyze dual solutions for Dual$_\sigma$, we'll start by showing that oblivious designs cannot achieve throughput better than $1/2$, even with high probability. Define
    \[
    \beta_{\sigma e} = 
    \begin{cases}
        2 & \mbox{if } e \mbox{ connects some } a\rightarrow \sigma(a) \mbox{ pair}  \\
        1 & \mbox{otherwise}.
    \end{cases}
    \]
and let $\alpha_{at\sigma}=\min_{P\in \pths_L(a,\sigma(a),t)}\{\sum_{e\in P} \beta_{\sigma e}\}$. By construction, $\alpha_{at\sigma} \ge 2$ for all $a,t$. So $\sum_{a,t} \alpha_{at\sigma} \ge 2 NT$, where $T$ is the period of the schedule.

Additionally, in expectation, $\mathbb{E} [\beta_{\sigma e}] = 1 + \frac{1}{N}$
for all $e$. So, $\mathbb{E} \left[ \sum_{ e} \beta_{\sigma e} \right] = (1 + \frac{1}{N}) N T$ \\
Then $\mathbb{E}[( \sum_e \beta_{\sigma e}) \big/ ( \sum_{at} \alpha_{at\sigma})] \leq \frac{1}{2}\left(1+\frac{1}{n}\right)$, which converges to $\frac{1}{2}$ as $N\rightarrow\infty$.

Now, suppose that throughput $r$ is achievable with high probability. That would mean that routing the demands $rD_\sigma$ gives a feasible flow with probability at least $\left(1-\frac{1}{N}\right)$ over a uniformly random choice $\sigma$. 
If routing demands $rD_\sigma$ is feasible for a fixed permutation $\sigma$, then it must be the case that the objective value of LP$_\sigma$ is at least $r$. 

And since the objective value of LP$_\sigma$ is always non-negative, then this implies that over a uniformly random permutation $\sigma$, the expected objective value of LP$_\sigma$ is at least $r\cdot(1-1/N)$.

The inequality $r\cdot(1-1/N) \leq \frac{1}{2}(1+\frac{1}{N})$ implies that $r$ must be at most $\frac{1}{2} + \frac{2}{N-1}$.

\textbf{Dual$_\sigma$ solutions to bound general throughput.}    
Now we'll show good dual solutions for general $r$. Given parameter $\theta\in \mathbb{Z}_{\geq 1}$, set 
    \[
    \beta_{\tau e} = 
    \begin{cases}
        \theta+1 & \mbox{if }e\mbox{ on a path of } \theta  \mbox{ phys edges between some }
         u\rightarrow\sigma(u) \mbox{ pair} \\
        1 & \mbox{otherwise}.
    \end{cases}
    \]
By construction, $\alpha_{at\sigma}\geq \theta +1$, so $\sum_{a,t}\alpha_{at\sigma} = NT(\theta +1)$, where $T$ is the period. Additionally,
\begin{align*}
    \mathbb{E} [\beta_{\sigma e}] &= 1 + \theta \Pr(e \mbox{ is on a path of $\le \theta$ phys edges with $\sigma$-matched endpoints})
\end{align*}
To bound the above value, we apply the following lemma from \cite{stoc-paper}.

\begin{restatable}{lemma}{lemcountinglem}\label{lem:counting-lem}
\textbf{(Counting Lemma)\cite{stoc-paper}} If in an ORN topology, some node $a$ can reach $k$ other nodes in at most $L$ timesteps using at most $h$ physical hops per path for some integer $h$, then $k\leq 2 {L \choose h}$, assuming $h\leq \frac{1}{3}L$.
\end{restatable}

Applying the Counting Lemma, the probability that edge $e$ is on a path of no more than $\theta$ physical edges with $\sigma$-matched endpoints is at most
\[ \frac{1}{N}\sum_{m=0}^{\theta-1}2{L\choose m}2{L\choose \theta-1-m} \leq \frac{4}{N}{2L\choose \theta-1} \]
assuming $\theta-1\leq\frac{1}{3}L$.
Then
\begin{align*}
    \mathbb{E} [\beta_{\sigma e}] & \leq 1 + \frac{4 \theta}{N} \binom{2L}{\theta-1} \\
    \implies \mathbb{E}\left[ \sum_e\beta{\sigma e} \right] & \leq NT\left(1 + \frac{4 \theta}{N} \binom{2L}{\theta-1}\right)
\end{align*}

Meaning we can bound the expected objective value of Dual$_\sigma$ throughput achievable under random permutation traffic.
\begin{align*}
    \E[\text{obj. value of Dual}_\sigma] &\leq \E\left[ \sum_e \beta_{\sigma e} \right] \big/ \left( NT(\theta+1) \right)\\
    &\leq \left( 1 + \frac{4 \theta}{N} \binom{2L}{\theta-1} \right) \big/ \left( \theta +1 \right) 
\end{align*}

As before, we use this expectation to find an upper bound on the achievable throughput rate with high probability
\[ r\left(1-\frac{1}{N^d}\right) \leq  \left( 1 + \frac{4 \theta}{N} \binom{2L}{\theta-1} \right) \big/ \left( \theta +1 \right) \]

We then simplify and isolate $L$ to one side of the inequality, to find the following lower bound on maximum latency. 
The inequality $\frac{a!}{(a-b)!}\leq a^b$ and Stirling's approximation $(k!)^{\frac{1}{k}} \geq \frac{k}{e}\sqrt{2\pi k}^{\frac{1}{k}}$ prove useful during this simplification process.
\[ L  \geq \frac{\theta-1}{2e} N^{\frac{1}{\theta-1}} \Bigg( \left(\frac{N^d-1}{N^d}r-\frac{1}{\theta+1}\right)\frac{\sqrt{2\pi(\theta-1)}}{4\theta} \Bigg) ^{\frac{1}{\theta-1}} \]

To ensure that this bound stays above 0, we approximately need $(r (\theta+1) - 1) > 0$, meaning $\theta$ must be greater than $\frac{1}{r}-1$. Setting $\theta$ as the smallest integer for which this holds, we find $\theta = \lfloor\frac{1}{r}\rfloor$.
Let $\semih=\theta-1$ and $\varepsilon = \semih + 1 - (\frac{1}{r}-1)$. Then we substitute $r=\frac{1}{\semih+2-\eps}$ to find

\begin{align*}
    L &\geq \frac{\semih}{2e} \left( \eps N \right)^{1/\semih} \left( \frac{\sqrt{2\pi \semih}}{4(h+1)(\semih+2-\eps)} \right)^{1/\semih}  \\
    \implies L &\geq \Omega\left( \semih\left((\varepsilon N)^{1/\semih} + N^{1/(\semih+1)}\right) \right).
\end{align*}
\end{proof}

\begin{restatable}{corollary}{corsemioblivtradeoff}\label{cor:semi-obliv-lb-maxlat}
Given any fixed throughput value $r\in(0,\frac{1}{2}]$, let $\semih = \semih(r) = \lfloor\frac{1}{r}-1\rfloor$ and $\eps = \eps(r) = \semih + 1 - (\frac{1}{r}-1)$. 
\begin{enumerate}
	\item Then any fixed SORN design which guarantees throughput $r$ (with respect to fixed demands), must suffer maximum latency at least $\Omega(L_{low}(r,N))$.
	\item Additionally, any distribution over SORN designs $\mathscr{S}$ each of size $N$, which guarantees throughput $r$ (with respect to fixed demands) over the random sampling $\mathcal{R}\sim\mathscr{R}$ must suffer at least $\Omega(L_{low}(r,N))$ maximum latency. 	
\end{enumerate}
\end{restatable}

Before we begin the proof, note that this lower bound does not make any claims about what maximum latencies are achievable with high probability for SORNs which guarantee throughput $r$. 
In \Cref{app:avg-lat-semi-obliv}, we give a similar lower bound on the {\em average} (or, expected) latency of any SORN design which guarantees throughput $r$. 
This lower bound has an additional multiplicative dependence on $\eps$.
Thus, the lower bound on maximum latency and the lower bound on expected latency
match to within a constant factor for most values of $r$: when
$ \frac{1}{r} \not\in \bigcup_{m = 2}^{\infty} \left( m - \frac{1}{K}, m \right) $, for any large constant $K$. 
\begin{proof}
The linear program as written in the proof of \Cref{thm:tradeoff}.4, when considered as whole instead of as a family of $N!$ different programs, sets up this SORN problem exactly. 
It asks: given a particular reconfigurable network schedule, for each possible permutation $\sigma$, maximize the guaranteed throughput rate while routing flow between $\sigma$-matched pairs. 
Since the expectation $\E_\sigma\left[ \sum_e \beta_{\sigma e} \right]$, is upper bound by $NT\left(1 + \frac{4 \theta}{N} \binom{2L}{\theta-1}\right)$, then there exists at least one $\sigma$ for which the bound holds. The rest of the proof follows similarly to the proof of \Cref{thm:tradeoff}.4, only without the factor $\left(1-\frac{1}{N^d}\right)$.

Note that every SORN design $\mathcal{S}$ within the support of $\mathscr{S}$ must itself guarantee throughput $r$ (with probability 1). 
Thus, each design $\mathcal{S}$ must suffer maximum latency at least $\Omega(L_{low}(r,N))$, and the whole distribution must also suffer maximum latency at least $\Omega(L_{low}(r,N))$.
\end{proof}

%% file: open-questions.tex
\section{Conclusion and Open Questions} \label{sec:open-questions}

In this paper, we showed that, compared to the guaranteed throughput versus latency tradeoff achieved in \cite{stoc-paper}, a strictly superior latency-throughput tradeoff is achievable when the throughput bound is relaxed to hold with high probability. 
We showed that the same improved tradeoff is also achievable with guaranteed throughput under time-stationary demands, provided the latency bound is relaxed to hold with high probability and that the network is allowed to be semi-oblivious, using an oblivious (randomized) connection schedule but demand-aware routing.
We proved that the latter result is not achievable by any fully-oblivious reconfigurable network design, marking a rare case in which semi-oblivious routing has a provable asymptotic advantage over oblivious routing.

\textbf{Removing the logarithmic gap and when $\eps$ is small.} 
Our designs only attain maximum latency $\bigo(L_{upp}(r,N))$ up to a $\bigotilde(\log N)$ factor, leaving a logarithmic gap between our upper and lower bounds. 
Is there an ORN or SORN design that achieves maximum latency $\bigo(L_{upp}(r,N))$?
Alternatively, is there a stronger lower bound than the one we presented in \Cref{sec:lower-bound}?

Additionally, when $\eps^{1/g}$ is sub-constant, then $L_{upp}(r,N) > \bigo(L_{low}(r,N))$.
This leaves us with a small but measurable fraction of throughput values for which we cannot find ORN and SORN designs which achieve provably optimal throughput-latency tradeoffs, even up to a logarithmic factor.
\cite{stoc-paper} handled this case by developing a second ORN family which sent flow on both $h$- and $(h+1)$-hop semi-paths.
We believe there a similar result for ORNs which achieve throughput with high probability and SORNs may be proven, by considering larger numbers of constellations when routing the hop-efficient paths.
However, we leave that to future work.

\textbf{Time-varying demands.}
In order to prove our throughput-latency tradeoffs for SORN designs, we were required to restrict ourselves to time-stationary (permutation) demands.
While this still shows that semi-oblivious routing has a provable asymptotic advantage over oblivious routing in the case of reconfigurable networks, it is desirable to find SORN designs which can handle time-varying demands.
Our SORN design $\mathscr{S}_N(g,C)$ works for almost all time-varying demands. 
However, in the case that it must route all flow (from every starting timestep $t$) along 2-hop paths, there is no obvious way to ``ramp back up'' to sending flow on $(g+1)$-hop paths again without waiting for most flow in the network to clear, which would require almost 2 full periods, or iterations of the schedule.

\textbf{Bridging the gap between theory and practice.}
As with previous work in this domain, we make several assumptions that do not hold for practical networks in order to make the analysis tractable.
In particular, our model of ORNs does not account for propagation delay between nodes.
In a practical network, it takes time for each message to traverse each physical link.
Our model of ORNs can easily be adjusted to take this into account with our definition of the virtual topology, 
and our design itself could be modified by taking advantage of the fact that flow paths always take at most one physical hop per phase block. 
However, large propagation delays penalize solutions which take more physical hops, which inherently changes the attainable throughput versus latency tradeoffs in a real system.
Once propagation delays become superlinear in $N$, one should always maximize throughput, since latency becomes dominated by propagation delay.
It is worth exploring where and how this shift from a full tradeoff curve to a single optimal point occurs, as propagation delay increases.

Additionally, we assume fractional flow: each unit of flow can be fractionally divided and sent across multiple different paths. In a practical network, flow is sent in discrete packets, which cannot be divided. 
Due to this assumption, our model sends small fractions of flow from multiple paths across the same link. However in a real system, only one packet from one path may traverse the link during a single timestep.
As a result, in real systems queuing may happen, which is best addressed using a congestion control system.
Congestion control has a decades-long history of active research across various networking contexts.
Our proposed designs present a new context for this area of research, and will likely require both adapting existing ideas from other contexts, as well as new innovations.

%% file: k-contentious-constellations.tex
\section{Mixing $(\semih+1)$-hop and 2-hop paths in our Semi-Oblivious Design} \label{app:k-frame-cor}

In defining the SORN design $\mathcal{S}$, we always chose to route permutation demand $D_\sigma$ on 2-hop paths if any edge $e$ in any $\chframe$ would become overloaded from routing $D_\sigma$ on $(\semih+1)$-hop paths.
However, this choice is a bit extreme. 
After all, the connection schedule of $\mathcal{S}$ iterates through $(p-1)^{\semih-1}$ different $\chframe$s.  

Label a $\chframe$ $A\mathcal{V}$ as {\em contentious} if there exists some edge $e$ occurring during constellation $A\mathcal{V}$ which is overloaded when routing demand $D_\sigma$ with the $(g+1)$-hop routing scheme.
It would be desirable the flow which would be routed along non-contentious constellations could still be routed on the more latency-efficient $(g+1)$-hop paths, while only the flow that would be routed on the contentious constellations is relegated to the 2-hop alternate paths.

This strategy slightly decreases the achievable throughput rate, due to reserving a small amount of edge capacity on each edge for 2-hop paths.
However, as long as the number of contentious $\chframe$s $k$ is small, we can still provably achieve throughput $r$ for any $r\in(0,\frac{1}{2})$ for which $\eps(r) = \floor{\frac{1}{r}-1}+1-(\frac{1}{r}-1)\neq 1$. (Or in other words, for $r$ which is not the reciprocal of an integer.) 

Specifically, if there are no more than $\frac{(1-\eps)(p-2)^g}{4(p-1)}$ contentious $\chframe$s over the entire period, then as described above, only route the flow that would be routed on contentious constellations on the alternate 2-hop paths.
(This is exactly the flow that originates during a constellation immediately prior to a contentious constellation.) 
Route all other flow on $(g+1)$-hop paths.
If there are more than $\frac{(1-\eps)(p-2)^g}{4(p-1)}$ contentious $\chframe$s, then route all flow on 2-hop paths.

\begin{restatable}{corollary}{corkframesdesign}\label{cor:k-frames-design}
Given a fixed throughput value $r$, let $\semih = \semih(r) = \lfloor\frac{1}{r}-1\rfloor$ and $\eps = \eps(r) = \semih + 1 - (\frac{1}{r}-1)$, and assume $\eps\neq 1$.
Let $\delta=\frac{1-\eps}{2(\semih-1)}$ and $C=\frac{6\log\log N}{\delta^2}\ln(N)$, and assume that $N=p^\semih$ for prime $p$ for which $C(\semih+1)\leq p$.
Consider the SORN design $\mathcal{S}$ described above with parameters $C$ and $\semih$, with the following alteration.

Then this scheme can guarantee throughput $r$ and achieves maximum latency $\bigotilde\left( gN^{1/g} \right)$ with high probability over the random sampling over $\sigma$, and achieves maximum latency $\bigotilde(N)$ in the low-probability case.
\end{restatable}

\begin{proof} 
Suppose that $k$ different $\chframe$s are contentious, and thus the flow which we would like to send only on $(\semih+1)$-hop paths within those frames must instead be sent on 2-hop paths across two iterations of the schedule.
This presents a balancing problem: since most $\chframe$s are not contentious, most of the edges this flow will be sent on have their own constellation's $(\semih+1)$-hop flows to forward along.
Thus, we need to bound the total amount of 2-hop flow on any edge in the network, given that $k$ different frame's worth of flow is being routed on 2-hop paths.

Fix an edge $e$, and consider for each demand pair $i,\sigma(i)$ how much much flow is crossing edge $e$ due to $i,\sigma(i)$. 
If 1st hop flow crosses edge $e$ from $i$ to $\sigma(i)$, then it must be the case that $tail(e)=i$ and both $head(e)-i$ and $\sigma(i)-head(e)$ have only non-zero coordinates.
Similarly, if 2nd hop flow crosses edge $e$ from $i$ to $\sigma(i)$, then $head(e)=\sigma(i)$ and both $tail(e)-i$ and $\sigma(i)-tail(e)$ have only non-zero coordinates.

If there are $k$ contentious $\chframe$s, then the total amount of flow that must be routed on 2-hop paths over the entire period will be $rkNC(g+1)(p-1)$, with each demand pair $i,\sigma(i)$ contributing $rkC(g+1)(p-1)$ flow per period.

For each edge $e=(a,b)$, consider the total amount of first-hop flow from 2-hop paths traversing the edge.
First-hop flow traversing $e$ must be traveling from source node $i=a$.
Also note that since edge $e$ exists in the network, then the vector $b-a$ must have only non-zero coordinates.
Then first-hop flow traverses edge $e$ only when $\sigma(a)-b$ also has only non-zero coordinates.

For node $a$, let us consider the set of other 2-hop paths which could carry flow from $a$ to $\sigma(a)$. 
(And thus, what other edges could carry first-hop 2-hop flow from $a$ to $\sigma(a)$.)
This is directly related to the number of nodes $b'$ for which $\sigma(a)-b'$ and $b'-a$ both have non-zero coordinates. 
This is at least $(p-2)^g$, which occurs exactly when $a$ and $\sigma(a)$ have no matching coordinates. 
Additionally, for a given first-hop edge $e$, the number of times an equivalent edge appears at any point in the period is the number of Vandermonde vectors in the constellation, or $C(g+1)$.

Thus, the amount of first-hop 2-hop flow that traverses edge $e$ is always no more than
\[ \frac{rkC(g+1)(p-1)}{(p-2)^gC(g+1)} = \frac{rk(p-1)}{(p-2)^g} . \]

A similar argument shows that the amount of second-hop 2-hop flow traversing edge $e$ will also be no more than $\frac{k(p-1)}{(p-2)^g}$.

Now that we have this bound, let us bound the total amount of $(g+1)$-hop and 2-hop flow traversing some edge $e$.

Fix an edge $e$ from a constellation that is not contentious. 
This edge will have both $(g+1)$-hop and 2-hop flow traversing it.
Since the constellation is not contentious, we know that the amount of $(g+1)$-hop flow traversing $e$ is no more than $(1+\delta)(g+1)r$.
Thus, the total amount of flow traversing edge $e$ is no more than 
\[ (1+\delta)(g+1)r + \frac{2rk(p-1)}{(p-2)^g} \] 

Setting this value equal to 1, thus maximizing $r$, we achieve

\begin{align*}
	(1+\delta)(g+1)r + \frac{2rk(p-1)}{(p-2)^g} & = 1 \\
	r \left( (1+\delta) (g+1) + \frac{2k(p-1)}{(p-2)^g} \right) & = 1
\end{align*}

Now replace $\delta = \frac{1-\eps}{2(g+1)}$ and $r=\frac{1}{g+2-\eps}$ and solve for $k$ to find the maximum value $k$ may take without overloading any edges.

\begin{align*}
	\frac{1}{g+2-\eps} \left( \left(1+ \frac{1-\eps}{2(g+1)}\right)(g+1) + \frac{2k(p-1)}{(p-2)^g} \right) & = 1 \\
	\left(1+ \frac{1-\eps}{2(g+1)}\right)(g+1) + \frac{2k(p-1)}{(p-2)^g} & = g + 2 - \eps \\
	g+1+\frac{1-\eps}{2} + \frac{2k(p-1)}{(p-2)^g} & = g+2-\eps \\
	\frac{2k(p-1)}{(p-2)^g} = \frac{1-\eps}{2} \\
	k = \frac{(1-\eps)(p-2)^g}{4(p-1)}
\end{align*}

As stated in the theorem statement, this is the maximum value $k$ can take without overloading edges in the network.

Now consider the probability that $k$ $\chframe$s are contentious. This is clearly no more than the probability that a single $\chframe$ is contentious, which occurs with high probability as stated in the proof of \Cref{thm:semi-low-prob}.

\end{proof}

%% file: avg-lat-lb.tex
\section{Average Latency Lower Bounds} \label{app:avg-lat-lb}

\subsection{Oblivious Designs} \label{app:avg-lat-obliv}

We devote this section to proving \Cref{thm:obliv-avg-lat-lb} as stated in \Cref{sec:semi-obliv-provable-sep}, restated below.

\thmlbinformal*

\begin{proof}

We begin by showing that the average latency of any fixed ORN design which guarantees throughput $r$  with respect to time-stationary demands must satisfy average latency at least $\Omega(L_{obl}(r,N))$.
This will be enough to proof \Cref{thm:obliv-avg-lat-lb}.
Note that every ORN design $\mathcal{R}$ within the support of $\mathscr{R}$ must guarantee throughput $r$ (with probability 1). 
Thus, each design $\mathcal{R}$ must satisfy average latency at least $\Omega(L_{obl}(r,N))$.
Use linearity of expectation to then show that the expected average latency of $\mathcal{R}\sim\mathscr{R}$ must be at least $\Omega(L_{obl}(r,N))$.

Fix any ORN connection schedule $\bm{\pi}$. We begin by stating the following linear program which, given $\bm{\pi}$ and average latency bound $L$, attempts to find a routing scheme which maximizes throughput, while keeping the average latency among all routing paths used, weighted by the fraction of flow traveling along each path, below the average latency bound $L$.

The proof will continue in the following way: we will first transform our LP into another LP which has fewer constraints. 
Then, we will take the Dual, to turn it into a minimization problem.
We will give a dual solution and upper bound its objective value, thus upper bounding guaranteed throughput subject to an average latency constraint.
Finally, we will rewrite this inequality into a lower bound on average latency, subject to a guaranteed throughput.

\begin{center}\fbox{\begin{minipage}{0.95\textwidth}
	\textbf{Primal LP}
	\vspace{-3mm}
	\[\begin{lparray}
    \mbox{maximize} & r \\
    \mbox{subject to} & \sum_{P\in\pths(a,b,t)} \mathcal{R}_{a,b,t}(P)  = r \hfill \forall a,b\in[N], t\in[T] \\
    & \sum_{a,t}\sum_{P\in\pths(a,\sigma(a),t) : e\in P} \mathcal{R}_{a,\sigma(a),t}(P)  \leq 1  \hfill  \qquad \forall e \in \ephys, \sigma \mbox{ permut on } [N] \\
    & \sum_{a,b,t}\sum_{P\in\pths(a,b,t)} \mathcal{R}_{a,b,t} \cdot lat(P) \leq rN^2T\cdot L \\
    & \mathcal{R}_{a,b,t}(P),r \geq 0 \hfill \forall a,b\in[N], t\in[T], P\in \pths(a,b,t)
  \end{lparray}
\] \end{minipage} }
\end{center}

Where we interpret $lat(P)$ as the latency of the path $P$, or the combined number of virtual and physical edges\footnote{We use $lat(P)$ here instead of $L(P)$ as in \Cref{sec:definitions} to denote the latency of path $P$ to prevent confusion between the latency of a path and the average latency bound $L$.}.
As in \cite{stoc-paper}, we replace the factorial number of constraints ranging over choices of $\sigma$ with a polynomial number of constraints which range over choices of $a,b\in[N]$. 
We do this by interpreting these constraints for a fixed edge $e$ as solving a maximum bipartite matching problem from $[N]$ to $[N]$. 
See Section 3.1 of \cite{stoc-paper} for a step-by-step explanation.

\begin{center}\fbox{\begin{minipage}{0.95\textwidth}
	\textbf{Primal LP}
	\vspace{-3mm}
	\[\begin{lparray}
    \mbox{maximize} & r \\
    \mbox{subject to} & \sum_{P\in\pths(a,b,t)} \mathcal{R}_{a,b,t}(P)  = r \hfill \forall a,b\in[N], t\in[T] \\
    & \sum_a\xi_{a,e} + \sum_b \eta_{b,e}  \leq 1  \hfill  \qquad \forall e \in \ephys \\
    & \sum_t\sum_{P\in\pths(a,\sigma(a),t) : e\in P} \mathcal{R}_{a,\sigma(a),t}(P) \leq \xi_{a,e} + \eta_{b,e} \hfill \forall a,b\in[N], e\in\ephys \\
    & \sum_{a,b,t}\sum_{P\in\pths(a,b,t)} \mathcal{R}_{a,b,t} \cdot lat(P) \leq r\cdot N^2TL \\
    & \mathcal{R}_{a,b,t}(P), \xi_{a,e}, \eta_{b,e},r \geq 0 \mbox{ }\hfill\mbox{ } \forall a,b\in[N], t\in[T], P\in \pths(a,b,t), e\in\ephys
  \end{lparray}
\] \end{minipage} }
\end{center}

\begin{center}\fbox{\begin{minipage}{0.95\textwidth}
	\textbf{Dual}
	\[\begin{lparray}
    \mbox{minimize} & \sum_{e} z_e \\
    \mbox{subject to} & \sum_{a,b,t} x_{a,b,t} - \gamma\cdot N^2TL \geq 1  \\
      & z_e \geq \sum_{b} y_{a,b,e} \hfill \text{ }\forall a\in[N],e\in\ephys \\
      & z_e \geq \sum_{a} y_{a,b,e} \hfill \text{ }\forall b\in[N],e\in\ephys \\
      & \sum_{e\in P} y_{a,b,e} + \gamma\cdot lat(P) \geq x_{a,b,t} \hfill \quad \qquad \forall a,b\in[N], \text{ }t\in[T],\text{ }P\in \pths(a,b,t) \\
      & y_{a,b,e},z_e,\gamma \geq 0 \hfill \text{ }\forall a,b\in[N], \text{ }e\in\ephys
  \end{lparray}
\]\end{minipage} }
\end{center}

We will first create a dual solution, aiming to fulfill all constraints except the first. 
We will then normalize the variables so that $\sum_{a,b,t} x_{a,b,t} - \gamma\cdot N^2TL$ is as close to 1 as possible.

We define this value $m^{+}_{\theta}(e,a)$ as follows.
\begin{align*}
m^{+}_{\theta}(e,a) & = \begin{cases}
  1 & \mbox{if } e \text{ can be reached from } a \text{ using at most } \theta \text{ physical hops }\\ 
  	& \hspace{6mm} \text{(including } e\text{) in } \leq kL \text{ timesteps} \\
  0 & \mbox{if } \text{ otherwise}
\end{cases}
\end{align*}
We define a similar value for edges which can reach node $b$.
\begin{align*}
m^{-}_{\theta}(e,b) & = \begin{cases}
  1 & \mbox{if } b \text{ can be reached from } e \text{ using at most } \theta \text{ physical hops }\\ 
  	& \hspace{6mm} \text{(including } e \text{) in } \leq kL \text{ timesteps} \\
  0 & \mbox{if } \text{ otherwise}
\end{cases}
\end{align*}

Set $\hat{y}_{a,b,e} = m^{+}_{\theta}(e,a) + m^{-}_{\theta}(e,b)$.
Also set $\hat{\gamma} = \frac{2\theta}{kL}$, and
set $\hat{x}_{a,b,t} = \min_{P\in \pths(a,b,t)} \{ \sum_{e\in P} \hat{y}_{a,b,e} + \hat{\gamma}\cdot lat(P)\}$. 
Note that by definition, $\hat{\gamma},\hat{x}$ and $\hat{y}$ variables satisfy the last set of dual constraints.

Consider some path $P$ which connects $a$ to $b$ starting at timestep $t$. 
If path $P$ has latency greater than $kL$, then 
\[\sum_{e\in P} \hat{y}_{a,b,e} + \hat{\gamma}\cdot lat(P) \geq \hat{\gamma}kL = 2\theta. \]

If on the other hand, path $P$ has latency no more than $kL$ but uses at least $\theta$ physical hops, then 
\[\sum_{e\in P} \hat{y}_{a,b,e} + \hat{\gamma}\cdot lat(P) \geq \sum_{e\in P} \hat{y}_{a,b,e} \geq 2\theta. \]

Finally, if path $P$ has latency no more than $kL$ and uses fewer than $\theta$ physical hops, then
\[ \sum_{e\in P} \hat{y}_{a,b,e} + \hat{\gamma}\cdot lat(P) \geq \sum_{e\in P} \hat{y}_{a,b,e} = 2|P\cap\ephys|. \]

We use the following lemma, restated from \Cref{sec:lower-bound}, to bound $\sum_{a,b,t}\hat{x}_{abt}$.

\lemcountinglem

\begin{align*}
	\sum_{a,b,t} \hat{x}_{a,b,t} & \geq \sum_{a,t}\sum_{b\neq a} \min \{2\theta, \min_{P\in \pths_{kL}(a,b,t)} \{2|P\cap\ephys|\}\} \\
	& \geq NT \left( 2\theta \left( N-2 {kL\choose \theta-1} \right) + 2{kL\choose\theta-1} \right) \\
	\implies \sum_{a,b,t} \hat{x}_{a,b,t} - \hat{\gamma}N^2TL & \geq NT \left( 2\theta \left( N-2 {kL\choose \theta-1} \right) + 2{kL\choose\theta-1} \right) - \frac{2\theta}{k}N^2T  \\
	 & = NT \left( \left(2\theta-\frac{2\theta}{k}\right) N - 4\theta {kL\choose \theta-1} + 4{kL\choose\theta-1} \right) = w
\end{align*}

Set this equal to $w$, our normalization term for each of the dual variables. Now set $\gamma = \frac{1}{w}\hat{\gamma}$, $y_{a,b,e} = \frac{1}{w}\hat{y}_{a,b,e}$ and $x_{a,b,t} = \frac{1}{w} \hat{x}_{a,b,t}$.

Finally, set $z_e = \max_{a,b} \{ \sum_{a}y_{a,b,e},\sum_{b}y_{a,b,e} \}$. 
Note that by construction, our dual solution satisfies all constraints.
To bounds throughput from above, we upper bound the sums $\sum_{a}y_{a,b,e}$ and $\sum_{b}y_{a,b,e}$, allowing us to upper bound the total sum of $z_e$ variables.

\begin{align*}
    \sum_{a}y_{a,b,e} & = \frac{1}{w} \sum_{a} \left(m^{+}_{\theta}(e,a) + m^{-}_{\theta}(e,b)\right)
      \leq \frac{1}{w}\left( \sum_{a} m^{+}_{\theta}(e,a) + N-1\right)
      \leq \frac{1}{w}\left(2{L\choose \theta-1} +N-1\right)
\end{align*}

where the last step is an application of the Counting Lemma. 
Similarly,

\begin{align*}
    \sum_{b}y_{a,b,e} & = \frac{1}{w} \sum_{b} \left(m^{+}_{\theta}(e,a) + m^{-}_{\theta}(e,b)\right)
     \leq \frac{1}{w}\left( N-1 + \sum_{b} m^{-}_{\theta}(e,b) \right)
     \leq \frac{1}{w}\left(N-1 + 2{L\choose \theta-1} \right)
\end{align*}

Recalling that $z_e = \max_{a,b} \{ \sum_{a}y_{a,b,e},\sum_{b}y_{a,b,e} \}$, and that the dual objective aims to minimize $\sum_e z_e$, we deduce that

\begin{align*}
	r \leq \sum_e z_e & \leq \frac{NT}{w}\left(N-1 + 2{kL\choose \theta-1} \right) \\
	& = \frac{N-1 + 2{L\choose \theta-1}}{\left(2\theta-\frac{2\theta}{k}\right) N - 4\theta {kL\choose \theta-1} + 4{kL\choose\theta-1}} \\
	& \leq \frac{ N-1 + 2\frac{(kL)!}{(\theta-1)!(kL-\theta+1)!} } { 2\theta \left( \left(\frac{k-1}{k}\right) N - 2 \frac{(kL)!}{(\theta-1)!(kL-\theta+1)!} \right)} \\
	& = \frac{k}{2\theta(k-1)} + \frac{4(kL)!} {2\theta(kL-\theta+1)! \left( \left(\frac{k-1}{k}\right)N(\theta-1)! - 2\frac{(kL)!}{(kL-\theta+1)!} \right)} \\
	& \leq \frac{k}{2\theta(k-1)} + \frac{4(kL)^{\theta-1}}{2\theta \left( \left(\frac{k-1}{k}\right)N(\theta-1)! - 2(kL)^{\theta-1} \right)}
\end{align*}

using the fact that $\frac{a!}{(a-b)!}\leq a^b$. 
At this point, we rearrange the inequality to isolate $L$.

\begin{align*}
	kL & \geq \left( \frac{\left(r-\frac{k}{2\theta(k-1)}\right) 2\theta \frac{k-1}{k}N (\theta-1)! } { 4\left(1+\theta\left(r-\frac{k}{2\theta(k-1)}\right)\right) } \right)^{\frac{1}{\theta-1}} \\
	L & \geq \frac{\theta-1}{ke} N^{\frac{1}{\theta-1}} \left( \frac{ \left(r-\frac{k}{2\theta(k-1)}\right) 2\theta \frac{k-1}{k} \sqrt{2\pi(\theta-1)} } { 4\left(1+\theta\left(r-\frac{k}{2\theta(k-1)}\right)\right) } \right)^{\frac{1}{\theta-1}}
\end{align*} 
using Stirling's approximation, in the form $(k!)^{\frac{1}{k}} \geq \frac{k}{e}\sqrt{2\pi k}^{\frac{1}{k}}$.

Recall that $h = \floor{\frac{1}{2r}}$ and $\eps_o = h+1 - \frac{1}{2r}$, as in the statement of the theorem above.
We also set the parameter $\theta=h+1$.
Note that our lower bound will always be positive when $\left(r-\frac{k}{2\theta(k-1)}\right) > 0$, which occurs as long as $\eps_o > \frac{h+1}{k}$. 
This tells us how to set the constant $k$: we may set $k = 2\frac{h+1}{\eps_o}$.
Since $\eps_o\in(0,1]$, this is always well-defined.
Substitute $h,\eps_o$ into the lower bound and simplify.

\begin{align*}
	L & \geq \frac{h}{ke} N^{\frac{1}{h}} \left( \frac{ \left(r-\frac{k}{2(h+1)(k-1)}\right) 2(h+1) \frac{k-1}{k} \sqrt{\pi h/2} } { 1+(h+1)\left(r-\frac{k}{2(h+1)(k-1)}\right) } \right)^{\frac{1}{h}} \\
	& = \frac{h}{ke} N^{\frac{1}{h}} \left( \sqrt{\frac{h\pi}{2}} \cdot \frac{ \left(\frac{1}{2(h+1-\eps_o)} - \frac{k}{2(h+1)(k-1)}\right) (h+1)\frac{k-1}{k} }{ 1+(h+1)\left(\frac{1}{2(h+1-\eps_o)} - \frac{k}{2(h+1)(k-1)}\right) } \right)^{\frac{1}{h}} \\
	& = \frac{h}{ke} N^{\frac{1}{h}} \left( \sqrt{\frac{h\pi}{2}} \cdot \frac{k-1}{k}\cdot \frac{k\eps_o - (h+1)}{3(h+1)(k-1) - 2\eps_o(k-1)} \right)^{\frac{1}{h}} \\
	& = \frac{h}{ke} (\eps_o N)^{\frac{1}{h}} \left( \sqrt{\frac{h\pi}{2}} \cdot \frac{k-1}{k}\cdot \frac{k-\frac{h+1}{\eps_o}}{3(h+1)(k-1) - 2\eps_o(k-1)} \right)^{\frac{1}{h}} \\
	& \geq \frac{h}{ke} (\eps_o N)^{\frac{1}{h}} \left( \sqrt{\frac{h\pi}{2}} \cdot\frac{1}{3(h+1)-2\eps_o} \right)^{\frac{1}{h}} \\
	& \geq \Omega\left( \frac{h}{k}(\eps_o N)^{\frac{1}{h}} \right) 
	 = \Omega\left( \eps_o (\eps_o N)^{\frac{1}{h}} \right) 
\end{align*}

because $\frac{1}{k} = \frac{\eps_o}{2(h+1)}$.
Finally, we realize that any lower bound on average latency subject to a guaranteed throughput constraint $r'<r$ is also a lower bound on average latency subject to guaranteed throughput $r$.
Let $r' = \frac{1}{2(h+1)}$. Then $r'<r$. 
Additionally, 
\[ \Omega\left( \eps_o(r') (\eps_o(r') N)^{\frac{1}{h(r')}} \right) =  \Omega\left(N^{\frac{1}{h+1}}\right) . \]

Therefore, combining these two lower bounds, we find that average latency of an ORN design which guarantees throughput $r$ must be at least
\[ \Omega\left( \eps_o(\eps_o N)^{\frac{1}{h}} + N^{\frac{1}{h+1}}\right) = \Omega\left( L_{obl}(r,N) \right) . \]

\end{proof}

\subsection{Semi-Oblivious Designs} \label{app:avg-lat-semi-obliv}

\begin{restatable}{theorem}{thmavglatsemilb}\label{thm:avg-lat-semi-lb}
Consider any constant $r \in (0,\frac{1}{2}].$ Let $\semih = \semih(r) = \floor{\frac{1}{r}-1}$ and $\eps = \eps(r) = \semih + 1 - (\frac{1}{r}-1)$,
	and let $L_{sem}(r,N)$ be the function 
	\[ L_{sem}(r,N) = \eps (\eps N)^{1/\semih} + N^{1/(\semih+1)} . \]
	Then for every $N > 1$ and every ORN design on $N$ nodes that achieves throughput $r$ with high probability, the average latency suffered by routing paths must be at least $\Omega(L_{sem}(r,N))$. 
\end{restatable}

\begin{proof}

We start by upper bounding throughput for a given average latency bound.
We begin with a set of $N!$ linear programs, one for each possible permutation $\sigma$ on the node set, to solve the following problem:
given an average latency bound $L$ and some reconfigurable network schedule, LP$_\sigma$ finds a set of routing paths to route $r$ flow between each $a,\sigma(a)$ pair, and maximizes the value $r$ for which this is possible.

\begin{center}\fbox{\begin{minipage}{0.95\textwidth}
	\textbf{Primal LP}
	\vspace{-3mm}
	\[\begin{lparray}
    \mbox{maximize} & r \\
    \mbox{subject to} & \sum_{P\in\pths(a,\sigma(a),t)} S_{\sigma}(a,t,P)  = r \hspace{23mm}\hfill \forall a\in[N], t\in[T], \sigma \mbox{ permut on } [N] \\
    & \sum_{a,t}\sum_{P\in\pths(a,\sigma(a),t) : e\in P} S_{\sigma}(a,t,P)  \leq 1  \hfill  \qquad \forall e \in \ephys, \sigma \mbox{ permut} \\
    & \sum_{a,t} \sum_{P\in\pths(a,\sigma(a),t) : e\in P} S_{\sigma}(a,t,P)\cdot lat(P) \leq r\cdot NTLN! \\
    &S_{\sigma}(a,t,P)  \geq 0 \hfill \forall a\in[N], t\in[T], \sigma \mbox{ permut}, P\in \pths(a,\sigma(a),t)
  \end{lparray}
\] \end{minipage} }
\end{center}

We then take the dual program of each LP to find the Dual program.

\begin{center}\fbox{\begin{minipage}{0.95\textwidth}
	\textbf{Dual$_\sigma$}
	\vspace{-3mm}
	\[\begin{lparray}
    \mbox{minimize} & \sum_{e,\sigma} \beta_{\sigma e} \\
    \mbox{subject to} & \sum_{a,t,\sigma}\alpha_{at\sigma} - \gamma NTLN! \geq 1 \\
    & \alpha_{at\sigma} \leq \sum_{e\in P} \beta_{\sigma e} + \gamma\cdot lat(P) \hspace{5mm}\hfill \forall a\in[N], t\in[T], \sigma \mbox{ permut}, P\in \pths(a,\sigma(a),t) \\
    & \gamma,\beta_{\sigma e} \geq 0 \hfill \forall e\in\ephys
  \end{lparray}
\]\end{minipage} }
\end{center}

For each permutation $\sigma$, we will define it's associated Dual variables. Then, we will analyze an upper bound on the objective value of the entire Dual program. 
We will also reframe the Dual program in the following way, which will be easier to work with. Note that $\left( \sum \beta_{\sigma e} \right) \big/\left( \sum \alpha_{at\sigma} - \gamma NTLN! \right)$ is still an upper bound on throughput.

\begin{center}\fbox{\begin{minipage}{0.60\textwidth}
	\vspace{-4mm}
	\[\begin{lparray}
    \mbox{minimize} & \left( \sum_{e,\sigma} \beta_{\sigma e} \right) \big/
        \left( \sum_{a,t,\sigma} \alpha_{at\sigma} - \gamma NTLN! \right) \\
    \mbox{subject to} & \alpha_{at\sigma} \leq \sum_{e\in P} \beta_{\sigma e} + \gamma\cdot lat(P) \hspace{5mm}\hfill \forall a,t,\sigma,P \\
    & \gamma,\beta_{\sigma e} \geq 0
  \end{lparray}
\]\end{minipage} }
\end{center}

Given parameter $\theta\in \mathbb{Z}_{\geq 1}$, set 
    \[
    \beta_{\sigma e} = 
    \begin{cases}
        \theta+1 & \mbox{if }e\mbox{ on a path of } \leq \theta  \mbox{ phys edges and } \leq kL \mbox{ latency } \\
        & \hspace{10mm}\mbox{between some }
         u\rightarrow\sigma(u) \mbox{ pair} \\
        1 & \mbox{otherwise}.
    \end{cases}
    \]

And set $\gamma = \frac{\theta+1}{kL}$. Then for any path $P$, $\sum_{e\in P} \beta_{\sigma e} + \gamma\cdot lat(P) \geq \theta+1$.
Therefore, we can always assign $\alpha_{at\sigma}=\theta+1$, giving us 
\[ \sum_{a,t,\sigma} \alpha_{at\sigma} - \gamma NTLN! = (\theta+1)NTN!\left(1 - \frac{1}{k}\right) \]

Finally, we upper bound $\E_\sigma[\sum_e \beta_{\sigma e}]$ to achieve an upper bound on $\sum_{e,\sigma} \beta_{\sigma e}$. 
We do this by upper bounding the expected value of the individual terms $\beta_{\sigma e}$.

\begin{align*}
    \mathbb{E}_\sigma [\beta_{\sigma e}] &= 1 + \theta \Pr[e \mbox{ is on a path of $\le \theta$ phys edges and $\leq kL$ lat. with $\sigma$-matched endpoints}]
\end{align*}

Applying the Counting Lemma (thus assuming $\theta-1\leq\frac{1}{3}L$), the above probability is at most
\begin{align*}
	\frac{1}{N}\sum_{m=0}^{\theta-1}2{kL\choose m}2{kL\choose \theta-1-m} \leq \frac{4}{N}{2kL\choose \theta-1}
\end{align*}

This is a sum over the number of physical hops $m$ taken before edge $e$. For each value $m$, we multiply the number of nodes $a$ which can reach edge $e$ using $m$ physical hops in latency no more than $kL$ by $\frac{1}{N}$ times the number of nodes $b$ reachable from $e$ using the remaining $(\theta-1-m)$ physical hops in latency no more than $kL$.
Then
\begin{align*}
    \mathbb{E}_{\sigma} [\beta_{\sigma e}] & \leq 1 + \frac{4 \theta}{N} \binom{2kL}{\theta-1} \\
    \implies \mathbb{E}_{\sigma}\left[ \sum_e\beta{\sigma e} \right] & \leq NT\left(1 + \frac{4 \theta}{N} \binom{2kL}{\theta-1}\right)
\end{align*}

Meaning we can bound the expected objective value of Dual$_\sigma$ throughput achievable under random permutation traffic.
\begin{align*}
    \E[\text{obj. value of Dual}] &\leq \E_\sigma\left[ \sum_{e} \beta_{\sigma e} \right] \big/ \left( NT(\theta+1)\left(\frac{k-1}{k} \right) \right)\\
    &\leq k\left( 1 + \frac{4 \theta}{N} \binom{2kL}{\theta-1} \right) \big/ \left( \theta +1 \right)(k-1) 
\end{align*}

Therefore, the guaranteed throughput rate of any SORN design must be
\[ r \leq k\left( 1 + \frac{4 \theta}{N} \binom{2kL}{\theta-1} \right) \big/ \left( \theta +1 \right)(k-1) \]

We then simplify and isolate $L$ to one side of the inequality, to find the following lower bound on maximum latency. 
The inequality $\frac{a!}{(a-b)!}\leq a^b$ and Stirling's approximation $(a!)^{\frac{1}{a}} \geq \frac{a}{e}\sqrt{2\pi a}^{\frac{1}{a}}$ prove useful during this simplification process.

\[ L \geq \frac{\theta-1}{2ke} N^{\frac{1}{\theta-1}} \left( \sqrt{2\pi(\theta-1)}\frac{ \frac{r(\theta+1)(k-1)}{k}-1 }{ 4\theta } \right)^{\frac{1}{\theta-1}} \]

To ensure the bound is positive, we need $\frac{r(\theta+1)(k-1)}{k}-1 > 0$, meaning that we need for $\theta > \frac{k}{(k-1)r}-1$, or approximately $\theta > \frac{1}{r}-1$. 
Setting $\theta$ as the smallest integer for which this holds, we find $\theta=\floor{\frac{1}{r}}$. 
Recall that $\semih = \semih(r) = \floor{\frac{1}{r}}-1$, therefore $\theta = \semih+1$. 
Additionally recall that $\eps = \eps(r) = \semih + 1 -(\frac{1}{r}-1)$.
We substitute $r = \frac{1}{\semih + 2 - \eps}$ and set $k = 2\frac{\semih+2}{\eps}$. Thus, the factor $\frac{1}{k}$ becomes $\frac{\eps}{2(g+2)}$, allowing the following lower bound on average latency to hold.

\begin{align*}
	L & \geq \frac{\semih}{2ke} (\eps N)^{1/\semih} \left( \sqrt{\frac{\pi \semih}{2}}\cdot \frac{k-\frac{\semih+2}{\eps}}{k(\semih+2-\eps)(\semih+1)} \right)^{1/\semih} \\
	\implies L & \geq \Omega\left( \eps(\varepsilon N)^{1/\semih}\right).
\end{align*}

Finally, we realize that any lower bound on average latency subject to a guaranteed throughput constraint $r'<r$ is also a lower bound on average latency subject to guaranteed throughput $r$.
Let $r' = \frac{1}{\semih+2}$. Then $r'<r$. 
Additionally, 
\[ \Omega\left( \eps(r') (\eps(r') N)^{\frac{1}{\semih(r')}} \right) =  \Omega\left(N^{\frac{1}{\semih+1}}\right) . \]

Therefore, combining these two lower bounds, we find that average latency of an ORN design which guarantees throughput $r$ must be at least
\[ \Omega\left( \eps(\eps N)^{\frac{1}{\semih}} + N^{\frac{1}{\semih+1}}\right) = \Omega\left( L_{sem}(r,N) \right) . \]

\end{proof}